\documentclass[11pt,a4paper,reqno]{amsart}%
\usepackage{amsthm,amsmath,amsfonts,amssymb,amsxtra,appendix,bookmark,dsfont}

\theoremstyle{plain}
\newtheorem{theorem}{Theorem}%[section]
\newtheorem{lemma}[theorem]{Lemma}

\theoremstyle{definition}

\theoremstyle{remark}
\newtheorem{remark}{Remark}

\DeclareMathOperator{\Tr}{Tr}

\def\leqslant{\le}
\def\bq{\begin{eqnarray}}
\def\eq{\end{eqnarray}}
\def\eqq{\end{align*}}
\def\bqq{\begin{align*}}
\def\nn{\nonumber}

\def\eps{\varepsilon}

\renewcommand{\epsilon}{\varepsilon}
\newcommand\1{{\ensuremath {\mathds 1} }}

\def\R {\mathbb{R}}
\def\C {\mathbb{C}}
\def\N {\mathcal{N}}

\def\E {\mathcal{E}}

\def\F {\mathcal{F}}

\def\H{\gH}

\def\R {\mathbb{R}}
\def\C {\mathbb{C}}
\def\N {\mathcal{N}}

\def\E {\mathcal{E}}

\def\d{{\rm d}}
\newcommand{\gH}{\mathfrak{H}}

\newcommand{\bH}{\mathbb{H}}
\newcommand{\dGamma}{{\ensuremath{\rm d}\Gamma}}
\title[Collective excitations of Bose gases]{Collective excitations of Bose gases in the mean-field regime}

\author[P.~T. Nam]{Phan Th\`anh Nam}
\address{IST Austria, Am Campus 1, 3400 Klosterneuburg, Austria} 
\email{pnam@ist.ac.at}

\author[R. Seiringer]{Robert Seiringer}
\address{IST Austria, Am Campus 1, 3400 Klosterneuburg, Austria} 
\email{robert.seiringer@ist.ac.at}

\begin{document}
\date{\today}

\begin{abstract} We study the spectrum of a large system of $N$ identical bosons interacting via a two-body potential with strength $1/N$. In this mean-field regime, Bogoliubov's theory predicts that the spectrum of the $N$-particle Hamiltonian can be approximated by that of an effective quadratic Hamiltonian acting on Fock space, which describes the fluctuations around a condensed state. Recently, Bogoliubov's theory has been justified rigorously in the case that the low-energy eigenvectors of the $N$-particle Hamiltonian display complete condensation in the unique minimizer of the corresponding Hartree functional. In this paper, we shall justify Bogoliubov's theory for the high-energy part of the spectrum of the $N$-particle Hamiltonian corresponding to (non-linear) excited states of the Hartree functional. Moreover, we shall  extend the existing results on the excitation spectrum to the case of non-uniqueness and/or degeneracy of the Hartree minimizer. In particular, the latter covers the case of rotating Bose gases, when the rotation speed is large enough to break the symmetry and to produce multiple quantized vortices in the Hartree minimizer.
\end{abstract}

\maketitle

\setcounter{tocdepth}{2}
\tableofcontents
%\addcontentsline{toc}{section}{Contents}

\section{Introduction}

We consider a system of $N$ identical bosons moving in an open subset $\Omega \subseteq \R^d$, described by the Hamiltonian 
\[
H_N= \sum\limits_{i = 1}^N T_{i} + \frac{1}{N-1} \sum\limits_{1 \leqslant i < j \leqslant N} {w(x_i-x_j)}
\]
which acts on the Hilbert space $\H^N=\bigotimes_{\text{sym}}^N L^2(\Omega)$ of permutation-symmet\-ric square-integrable functions. Here $T$ is the kinetic energy operator of each particle, and $T_i$ denotes the $N$-body operator $\1\otimes \cdots \otimes 1 \otimes T \otimes 1 \cdots \otimes 1$ where $T$ acts only on the $i$-th variable. The multiplication operator $w(.-.)$, with $w:\R^d \to \R$ an even, measurable function, corresponds to the interactions between particles. We consider the mean-field regime where the strength of the interaction is proportional to the inverse of the number of particles. 

We shall always assume that $T\ge 1$ and
\bq \label{eq:bound-w-T}
(w(.-.))^2 \le C_0 ( \1 \otimes T + T \otimes \1) \quad {\rm on}~\gH^2
\eq
for some constant $C_0>0$. For a physically relevant example, the reader may think of a system in $\Omega=\R^3$ with the kinetic energy operator $T=-\Delta+V(x)$ where $V_-:=\max\{-V,0\}$ is sufficiently regular ($|V_-|^2 \le -C\Delta$ is sufficient), and the Coulomb/Newton interaction $w(x-y)= \pm |x-y|^{-1}$. Under assumption (\ref{eq:bound-w-T}), it is easy to see that the interaction part in $H_N$ is infinitesimally small with respect to the kinetic energy part, and $H_N$ is a self-adjoint operator on $\bigotimes_{\text{sym}}^N D(T)$ by the Kato--Rellich Theorem \cite[Theorem X.13]{ReeSim2}. We are interested in the spectrum of $H_N$ when $N\to \infty$.

A crucial property of bosons is that a macroscopic fraction of particles can occupy a common quantum state. In Hartree's theory, one assumes that all particles live in a condensate state described by a normalized vector $u\in \gH = L^2(\Omega)$. In this case, the energy per particle is given by the {\em Hartree functional} 
$$\E_{\rm H}(u)=\frac{\langle u^{\otimes N}, H_N u^{\otimes N}\rangle}{N}=\langle u, T u \rangle+\frac12 \iint |u(x)|^2 w(x-y)|u(y)|^2\,\d x\,\d y,$$
which is well-defined on $D(T^{1/2})$, the quadratic form domain of $T$. In many situations (see e.g. \cite{FanSpoVer-80,PetRagVer-89,RagWer-89,LewNamRou-13}), Hartree's theory determines exactly the leading order of the ground state energy of $H_N$, namely $$\lim_{N\to \infty}\frac{\inf \sigma(H_N)}{N} = \inf \{ \E_{\rm H}(u) ~|~ \|u\|=1 \}=:e_{\rm H}\,.$$
 
Going beyond Hartree's theory, Bogoliubov's theory  \cite{Bogoliubov-47b} predicts the next order of the low-energy spectrum of weakly interacting Bose gases. In the mean-field regime, Bogoliubov's theory has been justified for the excitation spectrum of $H_N$ in the recent works \cite{Seiringer-11,GreSei-13,LewNamSerSol-13,DerNap-13}. It was shown that the low-energy eigenvalues of  $H_N-Ne_{\rm H}$ converge to those of the {\em Bogoliubov Hamiltonian}, an effective quadratic Hamiltonian on Fock space which is obtained by quantizing the Hessian of the Hartree functional at its minimizer.

In the present paper, we will justify Bogoliubov's theory for the {\em collective excitations of $H_N$}. Let $u_0$ be a stationary state of the Hartree functional; i.e., $u_0$ is a (normalized) solution to the Hartree equation
$$ (T + |u_0|^2*w -\mu_0) u_0 =0 $$
for some real constant $\mu_0$. If $u_0$ is a Hartree minimizer, then obviously it satisfies the above Hartree equation. However, in general $u_0$ is not necessarily a Hartree minimizer and $N \E_{\rm H}(u_0)$ may be very far from the ground state energy of $H_N$ (their distance is of order $N$). We will investigate the connection between the  eigenvalues of order $1$ of $H_N - N \E_{\rm H}(u_0)$ and those of the corresponding Bogoliubov Hamiltonian, which is the second quantization of the Hessian of the Hartree functional at $u_0$. A result of this kind was known only in the translation
invariant case \cite{Seiringer-11}, where collective excitations can be obtained from the low-energy spectrum using a Galileo transformation. In general, there is no such symmetry and our result is new. 

Our approach can also be used to extend the existing results on the low-energy excitation spectrum in \cite{Seiringer-11,GreSei-13,LewNamSerSol-13,DerNap-13} to the case of non-uniqueness and/or degeneracy (in the sense of absence of a non-zero lower bound to the Hessian of the Hartree functional) of the Hartree minimizer. 
 In particular, we shall consider the case of rotating Bose gases, where the system is rotation invariant with respect to a fixed axis. In this case the Hartree functional has infinitely many degenerate minimizers if the rotation speed is large enough to break the symmetry and to produce multiple quantized vortices in the Hartree minimizer.

\section{Main results}

In this section, we set up some notation and state our main results. All Hilbert spaces we consider are complex Hilbert spaces and their inner products are conjugate linear in the first variable and linear in the second. We always denote by $C$ a positive constant that depends only on the constant $C_0$ in assumption (\ref{eq:bound-w-T}) and the value $\langle u_0, T u_0 \rangle$, where $u_0$ is the relevant Hartree stationary state (two $C$'s in one line may refer to two different constants). 

To discuss Bogoliubov's theory, it is convenient to enlarge the $N$-particle space $\gH^N$ to the Fock space
$$ \F= \bigoplus_{m=0}^\infty \gH^m= \C \oplus \gH \oplus \gH^2 \oplus \cdots$$     
For every $f\in \gH$ one can define the annihilation operator $a(f)$ and the creation operator $a^\dagger(f)$, which are operators on $\F$ acting as
\begin{align*} (a(f) \Psi )(x_1,\dots,x_{m-1}) &= \sqrt{m} \int \overline{f(x_m)}\Psi(x_1,\dots,x_m) \d x_m \\
(a^\dagger (f) \Psi )(x_1,\dots,x_{m+1})&= \frac{1}{\sqrt{m+1}} \sum_{j=1}^{m+1} f(x_j)\Psi(x_1,\dots,x_{j-1},x_{j+1},\dots, x_{m+1})
\end{align*}
for every $\Psi\in \gH^m$. It is straightforward to see that $a^\dagger(f)$ is the adjoint of $a(f)$ and that they satisfy the canonical commutation relations (CCR)
$$ [a(f),a(g)]=[a^\dagger(f),a^\dagger(g)]=0,\quad [a(f), a^\dagger (g)]= \langle f, g \rangle$$
for all $f,g\in \gH$. 

We can extend the kinetic energy operator $\sum_{i=1}^N T_i$ on $\gH^N$ to an operator on Fock space 
$$ \dGamma( T) := 0 \oplus \bigoplus_{m=1}^\infty \left( \sum_{i=1}^m T_i \right).$$
If $\{u_n\}_{n=0}^\infty$ is an orthonormal basis for $\gH$, we can rewrite 
$$ \dGamma(T) = \sum_{m,n \ge 0} \langle u_m,  T u_n \rangle a_m^\dagger a_n$$
where we have denoted $a_n:=a(u_n)$ for short. In particular, the sum on the right side is independent of the choice of the basis.  Similarly, 
$$
0 \oplus 0 \oplus \bigoplus_{m=2}^\infty \left( \sum_{1\le i<j \le m} w(x_i-x_j) \right) = \frac{1}{2}\sum_{m,n,p,q\ge 0} W_{m,n,p,q} a_m^\dagger a_n^\dagger a_p a_q 
$$
where $W_{m,n,p,q}:= \langle u_m \otimes u_n, w\, u_p\otimes u_q \rangle$. Thus we can extend $H_N$ to an operator on Fock space
\bq \label{eq:2nd-quantization-HN}
\sum_{m,n \ge 0} \langle u_m,  T u_n \rangle a_m^\dagger a_n + \frac{1}{2(N-1)}\sum_{m,n,p,q\ge 0} W_{m,n,p,q} a_m^\dagger a_n^\dagger a_p a_q .
 \eq

In the following, we shall always choose $u_0\in D(T^{1/2})$ to be a stationary state of the Hartree functional, namely $u_0$ is a (normalized) solution to the Hartree equation
\bq \label{eq:Hartree-equation}
(T + |u_0|^2*w -\mu_0) u_0 =0 
\eq
for some real constant $\mu_0$ (which necessarily equals $\mu_0=\langle u_0, (T+|u_0|^2*w) u_0 \rangle$). Here $|u_0|^2*w$ is the convolution between the functions $|u_0|^2$ and $w:\R^d \to \R$. Let $P:=|u_0 \rangle \langle u_0|$ be the orthogonal projection onto $u_0$ and let $Q:=1-P$. Since the operator
$$ h:= T + |u_0|^2*w -\mu_0$$
leaves the subspace $\gH_+= Q \gH$ invariant, we shall often use the same notation for the restricted operator on $\gH_+$. 

Heuristically, Bogoliubov's theory \cite{Bogoliubov-47b} consists of the following approximation procedure. First, assuming most of the particles are in the condensate state $u_0$, we ignore all terms in (\ref{eq:2nd-quantization-HN}) which are higher than quadratic in $a_n$ and $a^\dagger_n$ with $n\ne 0$. Second, we replace $a_0$ and $a^\dagger_0$ by a scalar number $\sqrt{N_0}$ \footnote{Strictly speaking, the term $a_0^\dagger a_0^\dagger a_0 a_0 = a_0^\dagger a_0(a_0^\dagger a_0-1)$ is replaced by $N_0(N_0-1)$ instead of $N_0^2$.}, where $N_0$ is interpreted as the number of particles living in the condensate state. Finally, using $N_0\approx N$ and the Hartree equation (\ref{eq:Hartree-equation}), we formally arrive at
\bq \label{eq:Bogoliubov-approximation-0}
H_N - N \E_{\rm H}(u_0) \approx \bH\,,
\eq
where 
\bq \label{eq:Bogoliubov-Hamiltonian}
\bH= \dGamma( h+K_1) + \frac{1}{2}\sum_{m,n\ge 1} \Big( \langle u_m \otimes u_n, K_2 \rangle a_m^\dagger a_n^\dagger + \langle K_2,  u_m\otimes u_n \rangle a_m a_n \Big)\,.
\eq
Here $K_2\in \H^2$ is given by 
$$
K_2(x,y) := u_0(x) w(x-y) u_0(y)\,,
$$
and $K_1 = Q k_1 Q$, where $k_1$ is the operator defined via its integral kernel as 
$$
k_1(x,y):= u_0(x) w(x-y) \overline{u_0(y)} \,.
$$
The fact that $K_2\in \gH^2$ follows from (\ref{eq:bound-w-T}) and $u_0\in D(T^{1/2})$. Consequently, $k_1$  defines a Hilbert-Schmidt operator on $\gH$.

Since $h+K_1$ leaves $\gH_+$ invariant, the Bogoliubov Hamiltonian $\bH$ can be viewed as an operator on the excited Fock space
\bq \label{eq:excited-Fock-space}
\F_+= \C \oplus \gH_+ \oplus \gH_+^2 \oplus \cdots
\eq
On the other hand, $H_N$ is an operator on $\gH^N$. Therefore, the formal approximation (\ref{eq:Bogoliubov-approximation-0}) must be understood via an appropriate unitary transformation. Following \cite[Prop. 14]{LewNamSerSol-13}, we shall use the unitary mapping 
\begin{align}\label{eq:def-UN}
U_N : \gH^N   \quad &\to \quad \F_+^{\le N} := \C \oplus \gH_+ \oplus \cdots \oplus \gH_+^N \nn\\
\Psi \quad & \mapsto \quad \bigoplus_{j=0}^N Q^{\otimes j} \left( \frac{a_0^{N-j}}{\sqrt{(N-j)!}} \Psi \right),
\end{align}
which satisfies
\begin{align} \label{eq:alt-def-UN*}
U_N^\dagger : \F_+^{\le N} \quad &\to \quad \gH^N \nn\\
\bigoplus_{j=0}^N \phi_j \quad &\mapsto \quad \sum _{j = 0}^N\frac{{{{(a_0^*)}^{N - j}}}}{{\sqrt {(N - j)!} }}{\phi _j}
\end{align}
and, for all $m,n\ge 1$,
\begin{align*}
U_N \, a_0^\dagger a_0 \,U_N^\dagger &= N- \N_+ ,\\
U_N\, a^\dagger_m a_0 \,U_N^\dagger &= a^\dagger _m \sqrt{N-\N_+},\\
U_N\, a^\dagger_m a_n \,U_N^\dagger &= a^\dagger_m a_n
\end{align*}
on $\F_+^{\le N}$, where $\N_+=\dGamma(Q)$ is the particle number operator on $\F_+$. The unitary operator $U_N$ provides a tool to rigorously implement the c-number substitution in Bogoliubov's heuristic approximation. In fact, the approximate identity (\ref{eq:Bogoliubov-approximation-0}) should be understood as
\bq \label{eq:Bogoliubov-approximation-1}
U_N H_N U_N^\dagger -N \E_{\rm H}(u_0) \approx \bH .
\eq

We shall call a solution to the Hartree equation (\ref{eq:Hartree-equation}) {\em non-degenerate} if the Hessian of the Hartree functional $\E_{\rm H}(u)$ at $u_0$  is bounded from below by a strictly positive constant, i.e., if 
\bq \label{eq:non-degeneracy} 
\langle v, (h+K_1)v \rangle + {\rm Re} \iint \overline{v(x)v(y)} K_2(x,y) \d x \d y \ge \eta \|v\|^2 
\eq
for all $v\in \gH_+$ and for some constant $\eta >0$ independent of $v$.  As explained in \cite[Appendix A]{LewNamSerSol-13}, this non-degeneracy condition implies that the Bogoliubov Hamiltonian $\bH$ satisfies
\begin{align}\label{eq:non-degeneracy-bH}
\bH \ge \eta' \N_+ -C
\end{align}
for some $\eta'>0$ and $C>0$. 
In this case, $\bH$ is bounded from below and it can be properly defined as a self-adjoint operator on $\F_+$ using Friedrichs' method (see \cite[Theorem 1]{LewNamSerSol-13}).  

When $u_0$ is a non-degenerate Hartree minimizer, the approximation (\ref{eq:Bogoliubov-approximation-1}) was justified in \cite{LewNamSerSol-13} at the level of quadratic forms. In fact, the quadratic form estimates in \cite{LewNamSerSol-13} are sufficient to analyze the low-energy spectrum of $H_N$, thanks to the min-max principle. On the other hand, since we are interested in the highly excited part of the  spectrum of $H_N$ (when $u_0$ is not a Hartree minimizer) and the degenerate case (when $u_0$ is a degenerate Hartree minimizer), the following operator estimate will be more useful. 

\begin{theorem}[Operator bound]  \label{thm:operator-bound} Assume that \eqref{eq:bound-w-T} and \eqref{eq:Hartree-equation} hold and let $\bH$ be defined in (\ref{eq:Bogoliubov-Hamiltonian}). Then both $U_N H_N U_N^\dagger$ and $\1_{\F_+^{\le N}}\bH \1_{\F_+^{\le N}}$ can be extended to self-adjoint operators on $\F_+^{\le N}$ with the same domain 
$$D\Big(\1_{\F_+^{\le N}}\dGamma(QTQ)\1_{\F_+^{\le N}}\Big)\,,$$
and one has the operator inequality
$$ \Big( U_N H_N U_N^\dagger -N \E_{\rm H}(u_0) -\1_{\F_+^{\le N}}\bH \1_{\F_+^{\le N}} \Big)^2 \le \frac{C}{N} \Big(\dGamma(QTQ)\N_+^2 +1\Big).$$
\end{theorem}

\begin{remark}
Since the square root is operator monotone, the operator bound in Theorem \ref{thm:operator-bound} implies the quadratic form bound in \cite[Prop. 5.1]{LewNamSerSol-13}. On the other hand, note that our assumption \eqref{eq:bound-w-T} on $w^2$ is stronger than the corresponding assumption on $w$ in \cite{LewNamSerSol-13}.
\end{remark}

Theorem \ref{thm:operator-bound} will be proved in Section \ref{sec:operator-bound}. From Theorem \ref{thm:operator-bound} and a localization argument, we obtain

\begin{theorem}[Bogoliubov to N-body excitations]\label{thm:collective-excitation} Assume that \eqref{eq:bound-w-T} and \eqref{eq:Hartree-equation} hold and let $\bH$ be defined in (\ref{eq:Bogoliubov-Hamiltonian}). Assume that there exist $m\in \mathbb{N}$, $\lambda\in \mathbb{R}$ and orthonormal vectors $\{\Phi_j\}_{j=1}^m \subset \F_+$ satisfying 
$$\Phi_j \in \bigcap_{N\ge 1} D\Big(\1_{\F_+^{\le N}}\dGamma(QTQ)\1_{\F_+^{\le N}}\Big)$$ and 
$$(\bH-\lambda)\Phi_j=0$$
in the sense that $\1_{\F_+^{\le N}}(\bH-\lambda)\1_{\F_+^{\le N+2}}\Phi_j =0$ for all $N\in \mathbb{N}$. Let
$$\delta:= m  \max\big\{ 1, \lambda, \max_{1\le j \le m} \langle \Phi_j, \N_+ \Phi_j \rangle \big\}.$$
Then for every $N\ge 3\delta$, there exists an $\epsilon>0$ satisfying  
$$ \eps \le C \max \big\{ \delta^{1/2}N^{-1/6},\delta^{3/2}N^{-1/2} \big\}$$
such that the interval $(\lambda-\eps,\lambda+ \eps)$ contains (at least) $m$ eigenvalues (counting multiplicity), or 
one element of the essential spectrum, of $H_N-N\E_{\rm H}(u_0)$.
\end{theorem}

\begin{remark} A result similar to Theorem \ref{thm:collective-excitation} when $\lambda$ is in the essential spectrum of $\bH$ also holds true. In this case, the eigenvalue equations $(\bH-\lambda)\Phi_j=0$ have to be replaced by the approximate ones. To be precise, we have to replace each $\Phi_j$ by a sequence $\{\Phi_{j,\ell}\}_{\ell=1}^\infty$ and assume that 
$$ \langle \Phi_{j,\ell}, \Phi_{j',\ell} \rangle =\delta_{jj'}, \quad \lim_{\ell\to 0} \sup_{k\in \mathbb N} \| \1_{\F_+^{\le k}}(\bH-\lambda)\1_{\F_+^{\le k+2}}\Phi_{j,\ell} \|=0.$$
\end{remark}
Theorem \ref{thm:collective-excitation} will be proved in Section \ref{sec:collective-excitation}. It quantifies the extent to which the approximate Bogoliubov Hamiltonian $\bH$ predicts the part of the spectrum of  $H_N$ near the energy $N\E_{\rm H}(u_0)$. 
Obviously, the result in Theorem \ref{thm:collective-excitation} is only interesting if 
$$ \max_{1\le j \le m} \langle \Phi_j, \N_+ \Phi_j \rangle <\infty.$$ 
Such states $\Phi_j$ describe essentially pure condensates (with condensate wave function $u_0$) with only finitely many particles outside the condensate.

Roughly speaking, Theorem \ref{thm:collective-excitation} provides a mapping from the spectrum of $\bH$ to the spectrum of $H_N$. The reverse direction of this relation is more complicated. Note that, in general, one cannot expect that an eigenvalue of $H_N$ close to $N \E_{\rm H}(u_0)$ is related to the spectrum of the Bogoliubov Hamiltonian $\bH$ at all, unless the corresponding eigenfunction displays almost complete condensation in the state $u_0$. The following result gives a mapping from those eigenvalues of $H_N$, whose eigenvectors satisfy a certain condensation assumption, to the spectrum of $\bH$. 

\begin{theorem}[N-body to Bogoliubov excitations] \label{thm:many-body-collective-excitation} Assume that \eqref{eq:bound-w-T} and \eqref{eq:Hartree-equation} hold and let $\bH$ be defined in (\ref{eq:Bogoliubov-Hamiltonian}). Assume that for every $N$ large enough, there exists a wave function $\Psi_N\in \gH^N$ such that
$$ (H_N - N\E_{\rm H}(u_0)-\lambda_N)\Psi_N=0$$
and
$$ \langle \Psi_N, \N_+ \Psi_N \rangle + | \lambda_N| =o(N^{1/3}).$$
Then there exist normalized vectors $\Phi_N' \in \F_+^{\le N^{1/3}}$ such that
$$ \| (\bH-\lambda_N)\Phi_N' \|^2 \le C\frac{C\lambda_N^2 }{N^{2/3}} + \frac{C}{N^{1/3}}\langle \Phi_N, (\N_++1) \Phi_N \rangle. $$
In particular, $\| (\bH-\lambda_N)\Phi_N' \|\to 0$ as $N\to \infty$. Consequently, if $\lambda_N\to \lambda$, then $\lambda\in \sigma(\bH).$
\end{theorem}

Theorem \ref{thm:many-body-collective-excitation} will be proved in Section \ref{sec:many-body-collective-excitation}. The condition $\langle \Psi_N, \N_+ \Psi_N \rangle \ll N^{1/3}$ in Theorem \ref{thm:many-body-collective-excitation} is a technical assumption in our approach (similarly, in Theorem \ref{thm:collective-excitation} we also need $\delta \ll N^{1/3}$ to ensure $\eps \to 0$). It is presumably far from optimal, since complete Bose-Einstein condensation requires only that $\langle \Psi_N, \N_+ \Psi_N \rangle \ll N$. If one is only interested in the low-energy eigenvalues of $H_N$, then by using the min-max principle and a bootstrap argument as in \cite{LewNamSerSol-13}, one can show that the optimal assumption $\langle \Psi_N, \N_+ \Psi_N \rangle \ll N$ is indeed sufficient to ensure the convergence of the eigenvalues of $H_N$, provided that $u_0$ is a non-degenerate Hartree minimizer.
For the eigenvalues of $H_N$ which are far from the ground state energy, however, we do not know how to improve the condensation assumption $\langle \Psi_N, \N_+ \Psi_N \rangle \ll N^{1/3}$, even with the non-degeneracy condition as in \eqref{eq:non-degeneracy} or \eqref{eq:non-degeneracy-bH}. 

We note that the assumption on complete condensation is inevitable for the higher part of the spectrum. For example, if $w=0$ (i.e., the non-interacting case) and $T$ has compact resolvent, then for two   orthonormal eigenfunctions $u$ and $v$ of $T$,  $u^{\otimes N/2}\otimes_s v^{\otimes N/2}$ is an eigenfunction of $H_N$, but it does not display  complete condensation. 

Next we restrict our attention to the low-energy spectrum of $H_N$. The results in \cite{Seiringer-11,GreSei-13,LewNamSerSol-13,DerNap-13} are limited to the case when the Hartree functional has a unique, non-degenerate minimizer. However, it may happen that the Hartree functional has multiple minimizers, and one of the minimizer may be degenerate as well. For instance, in the case of attractive interactions a broken symmetry can lead to multiple minimizers \cite{AscFroGraSchTro-02,GuoSei-13}. Also in the repulsive case, symmetry breaking can occur in rotating systems due to the appearance of quantized vortices \cite{Sei-02,Sei-03,Aftalion-06,Fetter-09}.

If the Hartree functional has only finitely many minimizers and all of them are non-degenerate, then one may still follow the approach in \cite{LewNamSerSol-13} to obtain the convergence of the low-energy spectrum of $H_N$. The only difference is that in this case, we have many possible condensate states and the  whole union of the spectra of the corresponding Bogoliubov Hamiltonians contribute to the spectrum of $H_N$ at the large $N$ limit. 

\begin{theorem}[Excitation spectrum with multiple condensations] \label{thm:excitation-spectrum} 
Assume that $T$ has compact resolvent and that $T$ and $w$ satisfy \eqref{eq:bound-w-T}. Assume that the Hartree functional $\E_{\rm H}$ has exactly (up to a phase) $J$ (normalized) minimizers $\{f_j\}_{j=1}^J$ and all of them are non-degenerate. We associate each function $f_j$ with a Bogoliubov
 Hamiltonian $\bH_j$ acting on  $\F_{+j}=\bigoplus_{m=0}^\infty \bigotimes^m_{\rm sym} (\{f_j\}^\bot)$ in the same way as in (\ref{eq:Bogoliubov-Hamiltonian}), with $u_0$ replaced by $f_j$. Let $\mu_1(H_N) \le \mu_2(H_N) \le \dots$ be the eigenvalues of $H_N$. Let $\{\mu_\ell \}_{\ell=1}^\infty$ be the increasing sequence which is rearranged from the union (counting multiplicity) of the eigenvalues of the $\bH_j$'s. Then we have 
$$ \lim_{N\to \infty} \Big( \mu_\ell(H_N) - Ne_{\rm H}  \Big)= \mu_\ell$$
for every $\ell=1,2,\dots$, where $e_{\rm H}=\E_{\rm H}(f_j)$ is the Hartree ground state energy. 

Moreover, assume that $\mu_L<\mu_{L+1}=\mu_{L'}<\mu_{L'+1}$ for some $L'>L \ge 0$ (with the convention $\mu_0=-\infty$) and that the $L'-L$ numbers $\mu_{L+1},\dots,\mu_{L'}$ consist of $r_j \ge 0$ eigenvalues of $\bH_j$ (counting multiplicity) with the corresponding eigenvectors $\{\Phi_{j,i}\}_{i=1}^{r_j}$, for all $j\in \{1,2,\dots,J\}$. Then for every $\ell\in \{L+1,\dots,L'\}$, there is a subsequence of $\{\Psi_{N, \ell}\}_N$ (still denoted by $\{\Psi_{N, \ell}\}_N$ for short) satisfying
\begin{align} \label{eq:cv-eigenfunction-PsiN}
\lim_{N\to \infty} \Big\| \Psi_{N,\ell} - \sum_{j=1}^J \sum_{m=1}^{r_j} \theta_{j,m} U_{N,j}^\dagger \1_{\F_{+j}^{\le N}}\Phi_{j,m}\Big\| =0
\end{align}
where the $\theta_{j,m}$'s are complex numbers satisfying $ \sum_{j=1}^J \sum_{m=1}^{r_j} |\theta_{j,m}|^2 =1$. Here we have denoted by $U_{N,j}$ the unitary operators defined as in (\ref{eq:def-UN}) with $u_0$ replaced by $f_j$.
\end{theorem}

The non-degeneracy condition on the minimizers of $\E_{\rm H}$ 
 implies that there exists a constant $\eta>0$ such that
\bq \label{eq:BEC-fj-Hj}
\bH_j \ge \eta \N_{+j}-C
\eq
for $j=1,2,\dots,J$, where $\N_{+j}=\dGamma(\1_{ \{f_j\}^\bot} )$ is the number operator on $\F_{+j}$.
Eq.~(\ref{eq:BEC-fj-Hj}) ensures that each $\bH_j$ is bounded from below on $\F_{+j}$ and it has a unique ground state (see \cite[Theorem 1]{LewNamSerSol-13}). 

When $T$ has compact resolvent, the spectra of $H_N$ and the $\bH_j$'s are all discrete. Compactness of the resolvent of $T$ corresponds to the physical situation of trapped systems, when all particles are confined by an external potential. The result of Theorem~\ref{thm:excitation-spectrum} can be generalized and holds even in case $T$ does not have compact resolvent, as long as a suitable binding condition for the Hartree functional holds. In fact, the only place where we shall use the compactness of the resolvent of $T$ is to obtain the condensation of the low-energy eigenfunctions of $H_N$ in Hartree minimizers. This condensation is well-known for trapped systems \cite{FanSpoVer-80,PetRagVer-89,RagWer-89,LewNamRou-13} thanks to the quantum de Finetti Theorem \cite{Stormer-69,HudMoo-75}, but it also holds under more general conditions (see \cite[Theorem 1.1]{LewNamRou-13}).

If $J=1$, i.e., when the Hartree functional has a unique minimizer, the result in Theorem \ref{thm:excitation-spectrum} was already proved in \cite{LewNamSerSol-13}. To deal with the more general case, we utilize the quantum de Finetti theorem \cite{HudMoo-75} for trapped systems, as in \cite{FanSpoVer-80,PetRagVer-89,RagWer-89,LewNamRou-13}. It allows to split the $N$-body eigenfunctions into several components, which each component corresponding to a condensate in one Hartree minimizer. Each component can then be treated by adapting the method in \cite{LewNamSerSol-13}. We shall explain the details of the proof of Theorem \ref{thm:excitation-spectrum} in Section \ref{sec:multiple-Hartree}.  

Theorem~\ref{thm:excitation-spectrum} holds if the Hartree functional has finitely many non-degenerate minimizers. If $u_0$ is a {\em degenerate} Hartree minimizer, then  the corresponding Bogoliubov Hamiltonian $\bH$ may have infinitely many ground states. In this case, we can still apply Theorem \ref{thm:collective-excitation} to deduce that $\inf \sigma(\bH)$ is an accumulation point of the spectrum of $H_N- Ne_{\rm H} $ when $N\to \infty$, in the sense that for every $\eps>0$, the number of eigenvalues of $H_N - N e_{\rm H}$ (counting multiplicity) contained in $(-\eps,\eps)$ tends to infinity as $N\to \infty$. However, we cannot give a full analysis for the excitation spectrum of $H_N$ (in particular, we cannot show that $H_N-Ne_{\rm H}$ is bounded from below uniformly in $N$). We hope to be able to come back to this  problem in the future. 

If the system is invariant under a certain symmetry (e.g. rotations or translations), then the Hartree functional may have infinitely many Hartree minimizers related via that symmetry,  and the corresponding Bogoliubov Hamiltonians are all unitarily equivalent. In this case, the (common) ground state energy of the Bogoliubov Hamiltonian is also an accumulation point of the spectrum of $H_N$, due to Theorem \ref{thm:collective-excitation}. 

The proofs of the main theorems are given in the next sections.

\section{Proof of  Theorem \ref{thm:operator-bound}: Operator bound}\label{sec:operator-bound}

In this section, we give the proof of Theorem \ref{thm:operator-bound}. 

\begin{proof} A straightforward computation as in \cite[Eq. (44)]{LewNamSerSol-13} gives
\begin{align} \label{eq:bHN-decompose}
U_N H_N U_N^\dagger- N \E_{\rm H}(u_0) -\1^{\le N} \bH \1^{\le N} = \frac{1}{2}\sum_{j=0}^5 \Big( R_j + R_j^\dagger \Big) 
\end{align}
where we denote $\1_{\F_+^{\le N}}$ by $\1^{\le N}$ for short, and the $R_j$'s are operators on $\F_+^{\le N}$ defined as follows:
\begin{align*}
R_0 &= R_0^\dagger=  \frac1 2 W_{0000}\frac{\N_+(\N_+-1)}{N-1},\\
R_1&= R_1^\dagger= -\dGamma\left(Q(|u_0|^2*w)Q+K_1 \right) \frac{\N_+-1}{N-1}, \\
R_2&= -  2 a^\dagger (Q (|u_0|^2* w) u_0)\frac{\N_+ \sqrt{N-\N_+}}{N-1} , \\
R_3&= \1^{\le N} \sum_{m,n\ge 1} \langle u_m \otimes u_n, K_2 \rangle   a_m^\dagger a_n^\dagger \Big(  \frac{\sqrt{(N-\N_+ )(N-\N_+ -1)}}{N-1}-1 \Big),\\
R_4 & = R_4^\dagger= \frac{1}{2(N-1)} \sum_{m,n,p,q \ge 1} W_{mnpq} a_m ^\dagger a_n ^\dagger a_p a_q , \\
R_5 &= \frac{2}{N-1} \sum_{m,n,p \ge 1} W_{mnp0} a_m^\dagger a_n^\dagger a_p \sqrt{N-\N_+} .
\end{align*}
By the Cauchy-Schwarz inequality, we have the operator bound on $\F_+^{\le N}$ 
\bq \label{eq:operator-square} \Big( U_N H_N U_N^\dagger- N \E_{\rm H}(u_0) -\1^{\le N} \bH \1^{\le N} \Big)^2 \le 3 \sum_{j=0}^5 (R_j^\dagger R_j + R_j R_j^\dagger ).
\eq

Now we estimate carefully all terms on the right side of (\ref{eq:operator-square}). 
\\  
$\mathbf{j=0}$:
We can simply bound
\bq \label{eq:R0}
R_0^2 = \frac{1}{4}|W_{0000}|^2 \frac{\N_+^2(\N_+-1)^2}{(N-1)^2} \le \frac{C\N_+^4}{N^2}.
\eq
$\mathbf{j=1}$: Using the Cauchy-Schwarz inequality and assumption (\ref{eq:bound-w-T}), we find that
\begin{equation}
\label{u2w}
 (|u_0|^2*w)^2 \le |u_0|^2*(w^2) \le C T .
 \end{equation}
Moreover, $K_1$ is a Hilbert-Schmidt operator, and hence it is bounded. Therefore, the operator
$$X:=Q(|u_0|^2*w)Q+K_1$$
satisfies $X^2 \le C QTQ.$ Thus
\begin{align} \label{eq:bound-dQTQ}
(\dGamma (X))^2 &= \sum_{i} X_i ^2 + \sum_{i \ne j} X_i X_j 
\le \sum_{i} X_i ^2 + \sum_{i \ne j} \frac{X_i^2 Q_j + X_j^2 Q_i}{2} \nn\\
&\le C \sum_{i} (QTQ)_i + C \sum_{i \ne j} \frac{(QTQ)_i  Q_j + (QTQ)_j  Q_i}{2}\nn\\
&\le C \dGamma (QTQ) (\N_+ +1). 
\end{align}
Since $\dGamma(X)$ and $\N_+$ commute, we get
\begin{align} \label{eq:R1}
R_1^2  \le \frac{C}{N^2} \dGamma(QTQ) (\N_+ +1)^3 .
\end{align}
$\mathbf{j=2}$:  Because of (\ref{u2w}), the vector $v:=Q (|u_0|^2*w)u_0$ is bounded in $\H$. Using 
$$a(v) a^\dagger(v)= \|v\|^2 + a^\dagger(v) a(v) \le \|v\|^2 (\N_++1)$$
we obtain
\begin{align}\label{eq:R2*R2}
 R_2^\dagger R_2 &= 4 \frac{ \N_+ \sqrt{N-\N_+}}{N-1} a(v) a^\dagger(v)  \frac{ \N_+ \sqrt{N-\N_+}}{N-1} \nn \\
& \le 4 \|v\|^2  \frac{(\N_++1)\N_+^2 (N-\N_+)}{(N-1)^2} \le \frac{C \N_+^3}{N}.
\end{align}
On the other hand, since 
$$
R_2^\dagger = -2 \frac{ \N_+ \sqrt{N-\N_+}}{N-1} a(v) = -2 a(v) \frac{ (\N_+ -1) \sqrt{N-\N_++1}}{N-1}
$$
we find that
\begin{align}\label{eq:R2R2*}
R_2 R_2^\dagger  &= 4 \frac{ (\N_+ -1) \sqrt{N-\N_++1}}{N-1} a^\dagger(v) a(v) \frac{ (\N_+ -1) \sqrt{N-\N_++1}}{N-1} \nn \\
& \le 4 \|v\|^2  \frac{ \N_+ (\N_+ -1)^2 (N-\N_++1)}{(N-1)^2} \le \frac{C \N_+^3}{N}.
\end{align}
$\mathbf{j=3}$: We can rewrite
$$
R_3 = \mathbb{K_{\rm cr}} g(\N_+) = g(\N_+ -2)\mathbb{K_{\rm cr}}, \quad R_3^\dagger = g(\N_+) \mathbb{K}_{\rm cr} ^\dagger = \mathbb{K}_{\rm cr} ^\dagger  g(\N_+-2)\,,
$$
where
\begin{align} \label{eq:def-Kcr}
\mathbb{K_{\rm cr}}:= \sum_{m,n\ge 1} \langle u_m \otimes u_n, K_2 \rangle   a_m^\dagger a_n^\dagger
\end{align}
and
$$
g(\N_+):= \1_{\{\N_+\le N-2\}}\Big(  \frac{\sqrt{(N-\N_+ )(N-\N_+ -1)}}{N-1}-1 \Big).
$$
Note that if $\Phi_k \in \gH_+^k$ with $k\le N-2$, then 
\begin{align*} \mathbb{K}_{\rm cr} \Phi_k =  \frac{1}{\sqrt{k!(k+2)!}}\!\sum_{\sigma\in\mathfrak{S}_{k+2}}\!\!K_2(x_{\sigma(1)},x_{\sigma(2)})\,\Phi_k(x_{\sigma(3)},\dots,x_{\sigma(k+2)}),
\end{align*}
and hence 
\begin{align*}\| \mathbb{K}_{\rm cr} \Phi_k \|  \le  \sqrt{(k+1)(k+2)}\|K_2\|_{L^2} \| \Phi_k \|_{\gH_+^k} \le C(k+1)^2 \| \Phi_k\|_{\gH_+^k}
\end{align*}
by the triangle inequality. Therefore,
\bq \label{eq:Kcr*Kcr}
\mathbb{K}_{\rm cr} ^\dagger \mathbb{K}_{\rm cr} \le C (\N_+ +1)^2.
\eq
Consequently,
\begin{align} \label{eq:R3*R3}
R_3^\dagger R_3 = g(\N_+) \mathbb{K}_{\rm cr} ^\dagger \mathbb{K}_{\rm cr} g(\N_+) \le C (\N_+ +1)^2 g(\N_+)^2 \le \frac{C(\N_++1)^4}{N^2}.
\end{align}
Similarly, since
\begin{align} \label{eq:KcrKcr*}
\mathbb{K}_{\rm cr}\mathbb{K}_{\rm cr} ^\dagger   & = \Big(\sum_{m,n\ge 1} \langle u_m \otimes u_n, K_2 \rangle   a_m^\dagger a_n^\dagger \Big) \Big(\sum_{p,q\ge 1} \langle K_2, u_p \otimes u_q \rangle   a_p a_q \Big) \nn\\
&\le \Big( \sum_{m,n\ge 1} |\langle u_m \otimes u_n, K_2 \rangle|^2 \Big) \Big( \sum_{p,q\ge 1} a^\dagger_p a^\dagger_q a_p a_q \Big) \le \|K_2 \|^2_{L^2(\Omega^2)} \N_+^2 ,
\end{align}
we find that
\begin{align}\label{eq:R3R3*} R_3 R_3^\dagger &=  g(\N_+ -2) \mathbb{K}_{\rm cr}\mathbb{K}_{\rm cr} ^\dagger g(\N_+ -2) \le \frac{C(\N_++1)^4}{N^2}.
\end{align}
$\mathbf{j=4}$: On $\F_+^{\le N}$, we have
$$ R_4=\frac{1}{N-1} U_N \Big( \sum_{1\le i<j \le N} w'_{ij} \Big) U_N^\dagger \,,$$
where we have introduced the two-body operator 
$$w'=Q\otimes Q w(.-.) Q\otimes Q.$$
Therefore
$$ R_4^2=\frac{1}{(N-1)^2} U_N \Big(  \sum_{1\le i<j \le N}\sum_{1\le k<\ell \le N} w'_{ij} w'_{k\ell} \Big) U_N^\dagger. $$
From assumption (\ref{eq:bound-w-T}), get
$$(w')^2 \le C Q\otimes Q \Big( \1 \otimes T + T \otimes \1) Q\otimes Q = C \Big((QTQ) \otimes Q + Q \otimes (QTQ) \Big)$$
on $\gH^2$. Consequently,
\begin{equation} \label{eq:R4-a}\sum_{i\ne j} (w'_{ij})^2 \le \sum_{i\ne j}  C \Big((QTQ)_i  Q_j + Q_j (QTQ)_i \Big) 
\le C\dGamma(QTQ) \N_+.
\end{equation}
Using the Cauchy-Schwarz inequality and (\ref{eq:R4-a}), we have
\begin{align} 
\sum_{|\{i,j,\ell\}|=3} w'_{ij} w'_{j\ell} &= \sum_{|\{i,j,\ell\}|=3} (w'_{ij}Q_\ell) (w'_{j\ell}Q_i) \nn\\
 &\le \sum_{|\{i,j,\ell\}|=3} \frac{(w'_{ij})^2Q_\ell + (w'_{j\ell})^2Q_i}{2}  \nn\\
&\le C\dGamma(QTQ) \N_+^2, \label{eq:R4-b}\\
\sum_{|\{i,j,k,\ell\}|=4} w'_{ij} w'_{k\ell} &= \sum_{|\{i,j,k,\ell\}|=4} (w'_{ij}Q_k Q_\ell) (w'_{k\ell} Q_iQ_j) \nn\\
&\le \sum_{|\{i,j,k,\ell\}|=4}   \frac{(w'_{ij})^2 Q_k Q_\ell + (w'_{k\ell})^2Q_iQ_j}{2} \nn\\
&\le C\dGamma(QTQ) \N_+^3.\label{eq:R4-c}
\end{align}
Here, the sum with the constraint $|\{i,j,\ell,k\}|=4$ means the sum over all mutually different indices $i,j,k,\ell \in \{1,2,\dots,N\}$, and likewise for the case $|\{i,j,\ell\}|=3$. Combining (\ref{eq:R4-a}), (\ref{eq:R4-b}) and (\ref{eq:R4-c}), we find that
\bq \label{eq:R4}
R_4^2 \le \frac{C}{N^2} U_N (\dGamma(QTQ)\N_+^3) U_N^\dagger = \frac{C}{N^2} \dGamma(QTQ)\N_+^3.
\eq
$\mathbf{j=5}$: We can write 
$$ R_5 =\frac{2}{N-1} U_N \Big(   \sum_{1\le i<j \le N} w''_{ij}  \Big) U_N^\dagger $$
where we have introduced the two-body operator 
$$w''= Q\otimes Q w(.-.)  (P\otimes Q+ Q\otimes P).$$
Thus
\begin{align*}
R_5^\dagger R_5 = \frac{4}{(N-1)^2} U_N \Big(   \sum_{1\le i<j \le N}  \sum_{1\le k<\ell \le N} (w'')^\dagger_{ij} w''_{k\ell}  \Big) U_N^\dagger. 
\end{align*}
From assumption  (\ref{eq:bound-w-T}), we have 
\begin{align*} (w'')^\dagger w'' &= (P\otimes Q + Q \otimes P)w Q\otimes Q w (P\otimes Q + Q \otimes P) \\
&\le (P\otimes Q + Q \otimes P)w^2 (P\otimes Q + Q \otimes P) \\
& \le C \Big( \1 \otimes (QTQ) + (QTQ)\otimes \1 \Big)
\end{align*}
on $\gH^2$. Consequently,
\begin{align*} 
\sum_{i\ne j} (w'')^\dagger_{ij} w''_{ij}  & \le C \sum_{i\ne j}  \Big( (QTQ)_i + (QTQ)_j \Big) \le C N\dGamma(QTQ), \\
\sum_{|\{i,j,\ell\}|=3} (w'')^\dagger_{j\ell} w''_{ij}  &= \sum_{|\{i,j,\ell\}|=3} (( w'')^\dagger_{j\ell}Q_i) ( w''_{ij}Q_\ell) \\
&\le \sum_{|\{i,j,\ell\}|=3} \frac{ (w'')^\dagger_{ij} w''_{ij} Q_\ell + (w'')^\dagger_{j\ell} w''_{j\ell} Q_i }{2}\\
&\le C N \dGamma(QTQ) \N_+, \\
\sum_{|\{i,j,k,\ell\}|=4} (w''_{k\ell})^\dagger w''_{ij}  &= \sum_{|\{i,j,k,\ell\}|=4} ( (w'')^\dagger_{k\ell} Q_i Q_j) (w''_{ij}Q_k Q_\ell)  \\
&\le \sum_{|\{i,j,k,\ell\}|=4} \frac{(w'')^\dagger_{ij} w''_{ij} Q_k Q_\ell + (w'')^\dagger_{k\ell} w''_{k\ell} Q_i Q_j}{2} \\
&\le C N \dGamma(QTQ) \N_+^2
\end{align*}
where we have again used the same notation for the summation over mutually distinct indices. In summary,
$$
\sum_{1\le i<j \le N}  \sum_{1\le k<\ell \le N} (w'')^\dagger_{ij} w''_{k\ell}  \le C N \dGamma(QTQ) \N_+^2
$$
and hence
\begin{equation} \label{eq:R5*R5}
R_5^\dagger R_5  \le  \frac{C}{N} \dGamma(QTQ) \N_+^2.
\end{equation}

On the other hand, again by assumption  (\ref{eq:bound-w-T}), 
\begin{align*}
 w'' (w'')^\dagger &=Q\times Q w (P\otimes Q+ Q\otimes P)^2  w Q\otimes Q\\
&\le Q\times Q w^2 Q\otimes Q \le C \Big( Q\otimes (QTQ) + (QTQ)\otimes Q \Big)
\end{align*}
on $\gH^2$. Therefore, 
\begin{align*} 
\sum_{i\ne j}  w''_{ij} (w'')^\dagger_{ij} & \le C \sum_{i\ne j}  \Big((QTQ)_i Q_j  + Q_i(QTQ)_j \Big) \le C \dGamma(QTQ) \N_+, \\
\sum_{|\{i,j,\ell\}|=3}  w''_{j\ell} (w'')^\dagger_{ij} &= \sum_{|\{i,j,\ell\}|=3} \frac{ w''_{j\ell} (w'')^\dagger_{ij} + w''_{ij}(w'')^\dagger_{j\ell}  }{2} \\
&\le \sum_{|\{i,j,\ell\}|=3} \frac{ w''_{ij} (w'')^\dagger_{ij}  + w''_{j\ell} (w'')^\dagger_{j\ell}  }{2}\\
&\le C N \dGamma(QTQ) \N_+, \\
\sum_{|\{i,j,k,\ell\}|=4}  w''_{k\ell} (w'')^\dagger_{ij} &= \sum_{|\{i,j,k,\ell\}|=4} \Big( w''_{k\ell} (Q_i+Q_j)\Big)\Big( (w'')^\dagger_{ij} (Q_k+Q_\ell ) \Big) \\
&\le \sum_{|\{i,j,k,\ell\}|=4} \frac{ w''_{ij} (w'')_{ij}^\dagger (Q_k+Q_\ell) + w''_{k\ell} (w'')_{k\ell}^\dagger  (Q_i+Q_j) }{2} \\
&\le C N \dGamma(QTQ) \N_+^2.
\end{align*}
Thus
\begin{align} \label{eq:R5R5*}
 R_5 R_5^\dagger &= \frac{4}{(N-1)^2} U_N \Big(   \sum_{1\le i<j \le N}  \sum_{1\le k<\ell \le N}  w''_{k\ell} (w'')^\dagger_{ij}  \Big) U_N^\dagger \le  \frac{C}{N} \dGamma(QTQ) \N_+^2.
\end{align}
{\bf Conclusion.} We now collect all the bounds (\ref{eq:R0}), (\ref{eq:R1}), (\ref{eq:R2*R2}), (\ref{eq:R2R2*}), (\ref{eq:R3*R3}), (\ref{eq:R3R3*}),  (\ref{eq:R4}),  (\ref{eq:R5*R5}) and  (\ref{eq:R5R5*}). Using, in addition, that $\N_+ \le \dGamma(QTQ)$ and $\N_+\le N$ on $\F_+^{\le N}$, we obtain 
\begin{align} \label{eq:Rj*Rj-RjRj*} \sum_{j=0}^5 (R_j^\dagger R_j + R_j R_j^\dagger) \le \frac{C}{N} \Big(\dGamma(QTQ)\N_+^2 +1\Big)
\end{align}
on  $\F_+^{\le N}$. Inserting (\ref{eq:Rj*Rj-RjRj*}) into (\ref{eq:operator-square}), we find the desired operator inequality
$$ \Big( U_N H_N U_N^\dagger -N \E_{\rm H}(u_0) -\1_{\F_+^{\le N}}\bH \1_{\F_+^{\le N}} \Big)^2 \le \frac{C}{N} \Big(\dGamma(QTQ)\N_+^2 +1\Big).$$

From the obtained estimates, we can also easily deduce the self-adjointness of $U_N H_N U_N^\dagger$ and $\1_{\F_+^{\le N}} \bH \1_{\F_+^{\le N}}$. For the Bogoliubov Hamiltonian restricted to $\F_+^{\le N}$, we can write
$$
\1_{\F_+^{\le N}} \bH \1_{\F_+^{\le N}} = \dGamma(QTQ)+ \dGamma(h+K_1-QTQ)+ \1_{\F_+^{\le N}} ( \mathbb{K}_{\rm cr} +  \mathbb{K}_{\rm cr}^\dagger)\1_{\F_+^{\le N}}\,,
$$
with $\mathbb{K}_{\rm cr}$ defined in (\ref{eq:def-Kcr}). Similarly to (\ref{eq:bound-dQTQ}), we have
\begin{align} \label{eq:dGammah-dGammaQTQ}
(\dGamma(h+K_1-QTQ))^2 \le C \dGamma(QTQ)(\N_++1).
\end{align}
Moreover, from (\ref{eq:Kcr*Kcr}) and (\ref{eq:KcrKcr*}) we get
$$
\1_{\F_+^{\le N}} ( \mathbb{K}_{\rm cr} +  \mathbb{K}_{\rm cr}^\dagger)\1_{\F_+^{\le N}} ( \mathbb{K}_{\rm cr} +  \mathbb{K}_{\rm cr}^\dagger)\1_{\F_+^{\le N}} \le 
\1_{\F_+^{\le N}} ( \mathbb{K}_{\rm cr} +  \mathbb{K}_{\rm cr}^\dagger)^2\1_{\F_+^{\le N}} \le C(\N_+^2 +1).
$$
Thus by the Cauchy-Schwarz inequality, 
$$
\Big(\1_{\F_+^{\le N}} \bH \1_{\F_+^{\le N}}- \dGamma(QTQ) \Big)^2 \le \frac{1}{2} (\dGamma(QTQ))^2 + C (\N_++1)^2  
$$
on $\F_+^{\le N}$. Similarly, from (\ref{eq:Kcr*Kcr}), (\ref{eq:KcrKcr*}), (\ref{eq:Rj*Rj-RjRj*}) and (\ref{eq:dGammah-dGammaQTQ}) we find that
\begin{align} \label{eq:N-body-kinetic}
\Big(U_N H_N U_N^\dagger -N \E_{\rm H}(u_0)- \dGamma(QTQ) \Big)^2 \le \frac{1}{2} (\dGamma(QTQ))^2 + C (\N_++1)^2  
\end{align}
on $\F_+^{\le N}$. Since $\N_+$ is a bounded operator on $\F_+^{\le N}$, the  Kato--Rellich Theorem \cite[Theorem X.13]{ReeSim2} allows us to conclude that $\1_{\F_+^{\le N}} \bH \1_{\F_+^{\le N}}$ and $U_N H_N U_N^\dagger$ can be extended to self-adjoint operators on $\F_+^{\le N}$ with the domain $D(\1_{\F_+^{\le N}} \dGamma(QTQ) \1_{\F_+^{\le N}})$. 
\end{proof}

\section{Proof of Theorem \ref{thm:collective-excitation}: Bogoliubov excitations} \label{sec:collective-excitation}

We shall deduce Theorem  \ref{thm:collective-excitation} from the following

\begin{lemma}[Localization] \label{lem:localization} Assume that $(\bH -\lambda) \Phi=0$ (in the sense explained in Theorem~\ref{thm:collective-excitation}) for some normalized vector $\Phi\in \F_+$ and some $\lambda\in \mathbb{R}$. Let $f: \R \to [0,1]$ be a Lipschitz function such that $f(t)=1$ when $t\le 1/2$ and $f(t)=0$ when $t\ge 1$. Then for every $M \in [1,N]$,
\begin{align*} 
\Big\| \big(U_N H_N U_N^\dagger -N \E_{\rm H}(u_0) -\lambda \big) f \Big(\frac{\N_+}{M} \Big)\Phi \Big\|^2 \\
\le C \Big( \frac{M^2}{N}+\frac{\|f'\|_{\infty}^2}{M}\Big)  \max\{ 1, \lambda, \langle \Phi, \N_+ \Phi \rangle\}.
\end{align*}
\end{lemma}

\begin{proof} Let us assume $\langle \Phi, \N_+ \Phi \rangle<\infty$ (otherwise there is nothing to prove).  Denote $ \Phi^{\le M}:= f(\N_+/M)\Phi$ for short. Note that $\Phi^{\le M}\in \F_+^{\le M}\subset \F_+^{\le N}$. By the triangle and the Cauchy-Schwarz inequality,  
\begin{align}\label{eq:col-triangle}
& \| (U_N H_N U_N^\dagger -N \E_{\rm H}(u_0) -\lambda) \Phi^{\le M} \|^2 \nn\\
&\le 2\| (U_N H_N U_N^\dagger -N \E_{\rm H}(u_0) -\1_{\F_+^{\le N}} \bH) \Phi^{\le M} \|^2 +  2\| \1_{\F_+^{\le N}} (\bH-\lambda) \Phi^{\le M} \|^2 \nn \\
& \le 2\| (U_N H_N U_N^\dagger -N \E_{\rm H}(u_0) -\1_{\F_+^{\le N}} \bH) \Phi^{\le M} \|^2 +  2\| (\bH-\lambda) \Phi^{\le M} \|^2 .
\end{align}
We shall estimate each term of the right side of (\ref{eq:col-triangle}) separately. We shall denote $\1_{\F_+^{\le N}}$ by $\1^{\le N}$ for short.
\\
{\bf First term}. Theorem \ref{thm:operator-bound} implies that 
\begin{align} \label{eq:col-RS1-a} & \| (U_N H_N U_N^\dagger -N \E_{\rm H}(u_0) -\1^{\le N} \bH) \Phi^{\le M} \|^2 \nn\\
 &\le \frac{C}{N} \Big \langle \Phi^{\le M} , \big(\dGamma(QTQ)\N_+^2+1 \big)\Phi^{\le M} \Big\rangle \nn \\
&\le \frac{CM^2}{N} \Big \langle \Phi^{\le M} , \big(\dGamma(h+K_1)+ C\N_+ +1 \big) \Phi^{\le M} \Big\rangle .
\end{align}
In the last estimate we have used  that $\Phi^{\le M}\in \F_+^{\le M}$ as well as the bound $QTQ \le C(h+K_1 +C)$, which follows from assumption (\ref{eq:bound-w-T}). From the equation $(\bH-\lambda)\Phi=0$, we find that
\begin{align} \label{eq:dG=H-K}
\dGamma(h+K_1) \Phi^{\le M} =f \Big(\frac{\N_+}{M} \Big) \dGamma(h+K_1) \Phi = f \Big(\frac{\N_+}{M} \Big) (\lambda - \mathbb{K}) \Phi
\end{align}
where 
$$\mathbb{K}:=\mathbb{H}-\dGamma(h+K_1)=\frac{1}{2}(\mathbb{K}_{\rm cr}+ \mathbb{K}_{\rm cr}^\dagger)\,,$$
with $\mathbb{K_{\rm cr}}$ defined in (\ref{eq:def-Kcr}). 
From (\ref{eq:Kcr*Kcr}) and (\ref{eq:KcrKcr*}), we have
\begin{align} \label{eq:bK2}
\mathbb{K}^2 \le \frac{1}{2} ( \mathbb{K}_{\rm cr}^\dagger \mathbb{K}_{\rm cr} + \mathbb{K}_{\rm cr} \mathbb{K}_{\rm cr}^\dagger) \le C(\N_++1)^2.
\end{align}
Consequently, $|\mathbb{K}| \le C(\N_++1)$ and hence
\begin{align*}
& \Big \langle \Phi^{\le M} , \dGamma(h+K_1) \Phi^{\le M} \Big\rangle = \lambda \|\Phi^{\le M} \|^2 - \Big \langle f \Big(\frac{\N_+}{M} \Big)^2 \Phi,  \mathbb{K} \Phi \Big\rangle \\
&\le \lambda \|\Phi^{\le M} \|^2  +  \Big \langle f \Big(\frac{\N_+}{M} \Big)^2 \Phi,  C(\N_+ +1) f \Big(\frac{\N_+}{M} \Big)^2 \Phi \Big\rangle ^{1/2}  \Big \langle  \Phi,  C(\N_+ +1) \Phi \Big\rangle^{1/2} \\
&\le \lambda + C + C  \langle  \Phi,  \N_+ \Phi \rangle .
\end{align*}
Thus from (\ref{eq:col-RS1-a}) it follows that
\begin{align} \label{eq:col-RS1} &  \big\| (U_N H_N U_N^\dagger -N \E_{\rm H}(u_0)- \1^{\le N}\bH) \Phi^{\le M} \big\|^2 \nn\\
&\le \frac{CM^2}{N}  \Big \langle \Phi , (\lambda + C\N_+ +C )\Phi \Big\rangle \le \frac{CM^2}{N} \max\{1,\lambda, \langle \Phi, \N_+ \Phi \rangle\}.
\end{align}
{\bf Second term.} From the equation $(\bH-\lambda)\Phi=0$ we get
\bq \label{eq:col-RS2-a}
(\bH-\lambda) \Phi^{\le M} = \left[\bH, f \Big(\frac{\N_+}{M} \Big) \right] \Phi = \left[\mathbb{K}, f \Big(\frac{\N_+}{M} \Big) \right]  \Phi.
\eq
 We have
\begin{align*}
& 2\left[\mathbb{K}, f \Big(\frac{\N_+}{M} \Big) \right] = \left[\mathbb{K}_{\rm cr}, f \Big(\frac{\N_+}{M} \Big) \right] + \left[\mathbb{K}_{\rm cr}^\dagger, f \Big(\frac{\N_+}{M} \Big) \right] \\
& = \mathbb{K}_{\rm cr} \left( f \Big(\frac{\N_+}{M} \Big) - f \Big(\frac{\N_++2}{M} \Big)  \right) + \mathbb{K}_{\rm cr}^\dagger \left( f \Big(\frac{\N_+}{M} \Big) - f \Big(\frac{\N_+-2}{M} \Big)  \right).
\end{align*}
By the Cauchy-Schwarz inequality $(A+B)^\dagger (A+B)\le 2(A^\dagger A + B^\dagger B)$, we find that
\begin{align*}
& -4\left[\mathbb{K}, f \Big(\frac{\N_+}{M} \Big) \right]^2 = 4 \left[\mathbb{K}, f \Big(\frac{\N_+}{M} \Big) \right]^\dagger \left[\mathbb{K}, f \Big(\frac{\N_+}{M} \Big) \right] \\
&\le 2 \left( f \Big(\frac{\N_+}{M} \Big) - f \Big(\frac{\N_++2}{M} \Big)  \right)  \mathbb{K}_{\rm cr}^\dagger  \mathbb{K}_{\rm cr} \left( f \Big(\frac{\N_+}{M} \Big) - f \Big(\frac{\N_++2}{M} \Big)  \right) \\
& \quad + 2 \left( f \Big(\frac{\N_+}{M} \Big) - f \Big(\frac{\N_+-2}{M} \Big)  \right) \mathbb{K}_{\rm cr}\mathbb{K}_{\rm cr}^\dagger  \left( f \Big(\frac{\N_+}{M} \Big) - f \Big(\frac{\N_+-2}{M} \Big)  \right)\,.
\end{align*}
Using again $\mathbb{K}_{\rm cr}^\dagger \mathbb{K}_{\rm cr} + \mathbb{K}_{\rm cr} \mathbb{K}_{\rm cr}^\dagger \le C(\N_++1)^2$ and the operator bound
$$ \left( f\Big( \frac{\N_+ \pm 2}{M}\Big) -f\Big(\frac{\N_+}{M}\Big) \right)^2 \le \frac{4\|f'\|^2_{\infty}}{M^2}\1_{\F_+^{\le M+2}}, $$ 
we find that 
\begin{align} \label{eq:[K,f]}& - \left[\mathbb{K}, f \Big(\frac{\N_+}{M} \Big) \right] ^2 \le \frac{C\|f'\|^2_{\infty} }{M^2} (\N_+^2 + 1) \1_{\F_+^{\le M+2}} \le \frac{C\|f'\|^2_{\infty}  }{M}(\N_+ + 1) .
\end{align}
Thus we deduce from (\ref{eq:col-RS2-a}) that
\bq \label{eq:col-RS2}
\big\| (\bH-\lambda) \Phi^{\le M} \big\|^2  \le \frac{C\|f\|_{\infty}^2}{M} \langle \Phi, (\N_+ +1) \Phi \rangle.
\eq

Inserting (\ref{eq:col-RS1}) and  (\ref{eq:col-RS2}) into (\ref{eq:col-triangle}), we find the desired estimate. 
\end{proof}

\begin{remark} If we assume $\|\N_+ \Phi\|^2<\infty$ instead of $\langle \Phi, \N_+ \Phi \rangle<\infty$, then by following the above argument we obtain the improved estimate
\begin{align}  \label{eq:improved-rem1}
& \Big\| \big(U_N H_N U_N^\dagger -N \E_{\rm H}(u_0) -\lambda \big) f \Big(\frac{\N_+}{M} \Big)\Phi \Big\|^2 \nn\\
&\le C \Big( \frac{M}{N}+\frac{\|f'\|_{\infty}^2}{M^2}\Big)  \max\{ 1, |\lambda|\, \| \N_+ \Phi \|, \|\N_+ \Phi\|^2\}.
\end{align}
This estimate follows from the following refinements of (\ref{eq:col-RS1}) and (\ref{eq:col-RS2}), respectively:
\begin{align} \label{eq:col-RS1-refine} &  \big\| (U_N H_N U_N^\dagger -N \E_{\rm H}(u_0)- \1^{\le N}\bH) \Phi^{\le M} \big\|^2 \nn\\
&\le \frac{CM}{N} \max\{1,|\lambda| \| \N_+ \Phi \|, \| \N_+ \Phi \|^2 \},
\end{align}
and
\bq \label{eq:col-RS2-refine}
\big\| (\bH-\lambda) \Phi^{\le M} \big\|^2  \le \frac{C\|f\|_{\infty}^2}{M^2} \langle \Phi, (\N_+ +1)^2 \Phi \rangle.
\eq
In fact, similarly to (\ref{eq:col-RS2}), the bound (\ref{eq:col-RS2-refine}) follows immediately from (\ref{eq:col-RS2-a}) and (\ref{eq:[K,f]}). To obtain (\ref{eq:col-RS1-refine}), we can replace (\ref{eq:col-RS1-a}) by 
\begin{align*} 
&\| (U_N H_N U_N^\dagger -N \E_{\rm H}(u_0) -\1^{\le N} \bH) \Phi^{\le M} \|^2 \nn\\
& \le \frac{CM}{N} \Big \| ( \dGamma(h+K_1) +1) \Phi^{\le M} \Big\| \, \Big\| (\N_+ +1) \Phi^{\le M} \Big\|
\end{align*}
and then bound $\big \| \dGamma(h+K_1)  \Phi^{\le M} \big\|$ using \eqref{eq:dG=H-K} and \eqref{eq:bK2}. 
\end{remark}

Now we are able to give the

\begin{proof}[Proof of Theorem  \ref{thm:collective-excitation}] Let us assume that $\max_{1\le j \le m} \langle \Phi_j, \N_+ \Phi_j \rangle <\infty$ (otherwise, there is nothing to prove). 
Fix a Lipschitz function $f:\R\to [0,1]$ as in Lemma \ref{lem:localization} (we can take $f(t)=\1_{(-\infty,1/2]} (t) + 2(1-t)\1_{[1/2,1]}(t)$). Then for every $M\in [1,N]$ and $j\in \{1,2,\dots,m\}$, from Lemma \ref{lem:localization} one has
\begin{align}\label{eq:col-tradiction-LS} &\quad \Big\| \big( U_N H_N U_N^\dagger -N \E_{\rm H}(u_0) -\lambda \big) f\Big(\frac{\N_+}{M} \Big)\Phi_j \Big\|^2 \nn
\\
&\le C  \Big( \frac{M^2}{N}+ \frac{1}{M}\Big)\max\big\{ 1, \lambda, \langle \Phi_j, \N_+ \Phi_j \rangle \big\}.
\end{align}
Moreover, since $f(t)=1$ when $t\le 1/2$,
\begin{align}\label{eq:col-tradiction-orthogonal} &\quad \left| \delta_{jk} - \left\langle f\Big( \frac{\N_+}{M}\Big) \Phi_j, f\Big( \frac{\N_+}{M}\Big) \Phi_k \right \rangle  \right| = \left| \left\langle \Phi_j, \Big( \1-  f^2\Big( \frac{\N_+}{M}\Big) \Big) \Phi_k \right \rangle\right| \nn\\
&\le \max_{i} \left| \left\langle \Phi_i, \Big( \1-  f^2\Big( \frac{\N_+}{M}\Big) \Big) \Phi_i \right \rangle\right| \le \frac{2}{M}\max_{i} \langle \Phi_j, \N_+ \Phi_j \rangle .
\end{align}

The conclusion then follows from \eqref{eq:col-tradiction-LS}, \eqref{eq:col-tradiction-orthogonal} and a standard technique in spectral theory. To be precise, let us take $\eps>0$ and assume that the interval $(\lambda-\eps,\lambda+\eps)$ contains no point of the essential spectrum and fewer than $m$ eigenvalues (counting multiplicity) of the (self-adjoint) operator $H_N-N \E_{\rm H}(u_0) $ on $\gH^N$. Since $U_N:\gH^N \to \F_+^{\le N}$ is a unitary transformation,  there exist orthonormal vectors $\{v_k\}_{k=1}^{m-1}\subset \F_+^{\le N}$ such that 
\bq \label{eq:col-contradiction}
 \Big( U_N H_N U_N^\dagger -N \E_{\rm H}(u_0) -\lambda \Big)^2  \ge \eps^2  \Big( \1 -\sum_{k=1}^{m-1}  |v_k \rangle \langle v_k| \Big).
 \eq
Denote 
$$ \delta:= m  \max\big\{ 1, \lambda, \max_{1\le j \le m} \langle \Phi_j, \N_+ \Phi_j \rangle \big\}.$$
Combining \eqref{eq:col-tradiction-LS}, \eqref{eq:col-tradiction-orthogonal} and \eqref{eq:col-contradiction}, we have
\begin{align}\label{eq:col-semi-final}
C \delta \Big( \frac{M^2}{N}+ \frac{1}{M}\Big) &\ge \sum_{j=1}^m \Big\| \big( U_N H_N U_N^\dagger -N \E_{\rm H}(u_0) -\lambda \big) f\Big(\frac{\N_+}{M} \Big)\Phi_j \Big\|^2 \nn\\
& \ge\eps^2 \sum_{j=1}^m  \left( \Big\| f \Big(\frac{\N_+}{M} \Big) \Phi_j \Big\|^2 - \sum_{k=1}^{m-1} \Big|\Big\langle v_k, f \Big(\frac{\N_+}{M} \Big) \Phi_j \Big\rangle\Big|^2 \right)\nn\\
& = \eps^2 \left( \sum_{j=1}^m   \Big\| f \Big(\frac{\N_+}{M} \Big) \Phi_j \Big\|^2 -\sum_{k=1}^{m-1} \sum_{j=1}^m \Big|\Big\langle f \Big(\frac{\N_+}{M} \Big) v_k,  \Phi_j \Big\rangle\Big|^2 \right) \nn\\
&\ge \eps^2 \left( \sum_{j=1}^m   \Big\| f \Big(\frac{\N_+}{M} \Big) \Phi_j \Big\|^2 -\sum_{k=1}^{m-1} \Big\| f \Big(\frac{\N_+}{M} \Big) v_k \Big\|^2 \right) \nn\\
& \ge \eps^2 \left( \sum_{j=1}^m   \Big\| f \Big(\frac{\N_+}{M} \Big) \Phi_j \Big\|^2 -m+1\right) \nn\\
& \ge \eps^2 \Big( 1-\frac{2\delta}{M} \Big) .
\end{align}
Here we have used the orthonormality of $\{\Phi_j\}_{j=1}^m$ and the fact that $|f|\le 1$. By choosing $M=\max\{3\delta,N^{1/3}\}$, we obtain
\bq \label{eq:col-final}
 \eps \le C \max \big\{ \delta^{1/2} N^{-1/6}, \delta^{3/2} N^{-1/2} \big\}
 \eq
and the desired conclusion follows.
\end{proof}

\begin{remark} If $\delta':=m\max\{1,\lambda, \max_{1\le j \le m}\|\N_+\Phi_j\|\}<\infty,$
then using the refinement in \eqref{eq:improved-rem1} we can improve (\ref{eq:col-semi-final}) to
\bq \label{eq:col-semi-final2}
C (\delta')^2 \Big( \frac{M}{N}+ \frac{1}{M^2}\Big) \ge \eps^2 \Big( 1-\frac{4(\delta')^2}{M^2} \Big) .
\eq
By choosing $M=\max\{3\delta',N^{1/3}\}$, this yields the improved bound 
$$\eps \le C \max\{\delta' N^{-1/3}, (\delta')^{3/2} N^{-1/2}\}.$$
\end{remark}

\section{Proof of Theorem \ref{thm:many-body-collective-excitation}: N-body excitations} \label{sec:many-body-collective-excitation}

In order to prove Theorem \ref{thm:many-body-collective-excitation}, we shall need the following analogue of (\ref{eq:[K,f]}).

\begin{lemma}[Commutator estimate]\label{le:N-body-commutator} Let $f:\R\to [0,1]$ be a Lipschitz function as in Lemma \ref{lem:localization}. Then for every $M\in [1,N]$, one has
\begin{align*} - \left[ U_N H_N U_N^\dagger, f\Big(\frac{\N_+}{M} \Big) \right]^2 \le \frac{C\|f'\|_{\infty}^2}{M^2} \Big( \frac{1}{N} \dGamma(QTQ)\N_+^2 + (\N_++1)^2 \Big) \1^{\le M+2}
\end{align*}  
on $\F_+^{\le N}$, where we have again denoted $\1_{\F_+^{\le k}}$ by $\1^{\le k}$ for short.
\end{lemma}
\begin{remark} \label{re:HN-bHN-commutator} Since $U_N \N_+ U_N^\dagger =\N_+$, the inequality in Lemma \ref{le:N-body-commutator} is equivalent to the following operator inequality on $\gH^N$
\begin{align*} - \left[ H_N, f\Big(\frac{\N_+}{M} \Big) \right]^2 \le \frac{C\|f'\|_{\infty}^2}{M^2} \Big( \frac{1}{N} \dGamma(QTQ)\N_+^2 + (\N_++1)^2 \Big) U_N^\dagger \1^{\le M+2} U_N.
\end{align*} 
\end{remark}

\begin{proof}[Proof of Lemma \ref{le:N-body-commutator}]  Denote
\begin{equation}\label{eq:notH}
 \bH_N:= U_N H_N U_N^\dagger - N\E_{\rm H}(u_0) \,.
\end{equation}
We shall use again the decomposition as in (\ref{eq:bHN-decompose}), namely
$$ \bH_N=\dGamma(h+K_1)+ \frac{1}{2} \1^{\le N}( \mathbb{K}_{\rm cr} + \mathbb{K}_{\rm cr} ^\dagger) + \frac{1}{2}\sum_{j=0}^5 \Big( R_j + R_j^\dagger \Big) \,,
$$
with the pair-creation operator $\mathbb{K}_{\rm cr}$ defined in (\ref{eq:def-Kcr}). We have
\begin{align*} \left[ \bH_N, f\Big(\frac{\N_+}{M} \Big) \right] &= \frac{1}{2} \sum_{j=2,3,5} \left[ R_j + R_j^\dagger , f\Big(\frac{\N_+}{M} \Big) \right]  + \frac{1}{2} \left[ \1^{\le N} ( \mathbb{K}_{\rm cr} + \mathbb{K}_{\rm cr} ^\dagger) , f\Big(\frac{\N_+}{M} \Big) \right]\\
& = \frac{1}{2} (R_2+R_5) \left( f\Big(\frac{\N_+}{M} \Big) - f\Big(\frac{\N_+ +1}{M} \Big) \right) \\
&\quad + \frac{1}{2} (R_2^\dagger + R_5^\dagger  ) \left( f\Big(\frac{\N_+}{M} \Big) - f\Big(\frac{\N_+ -1}{M} \Big) \right)  \\
&\quad  + \frac{1}{2} (R_3 + \mathbb{K}_{\rm cr} \1^{\le N-2})  \left( f\Big(\frac{\N_+}{M} \Big) - f\Big(\frac{\N_+ +2}{M} \Big) \right) \\
&\quad + \frac{1}{2} ( R_3^\dagger + \mathbb{K}_{\rm cr}^\dagger) \left( f\Big(\frac{\N_+}{M} \Big) - f\Big(\frac{\N_+ -2}{M} \Big) \right) .
\end{align*}
With the aid of the Cauchy-Schwarz inequality, we can bound
\begin{align}\nn 
& -\left[ \bH_N, f\Big(\frac{\N_+}{M} \Big) \right]^2 = \left[ \bH_N , f\Big(\frac{\N_+}{M} \Big) \right]^\dagger \left[ \bH_N, f\Big(\frac{\N_+}{M} \Big) \right]  \\ \nn 
&\le 2 \left( f\Big(\frac{\N_+}{M} \Big) - f\Big(\frac{\N_+ +1}{M} \Big) \right) (R_2^\dagger R_2 + R_5^\dagger R_5 )   \left( f\Big(\frac{\N_+}{M} \Big) - f\Big(\frac{\N_+ +1}{M} \Big) \right) \\ \nn 
&\quad + 2\left( f\Big(\frac{\N_+}{M} \Big) - f\Big(\frac{\N_+ -1}{M} \Big) \right) ( R_2 R_2^\dagger + R_5 R_5^\dagger ) \left( f\Big(\frac{\N_+}{M} \Big) - f\Big(\frac{\N_+ -1}{M} \Big) \right)\\ \nn
&\quad + 2 \left( f\Big(\frac{\N_+}{M} \Big) - f\Big(\frac{\N_+ +2}{M} \Big) \right) (R_3 ^\dagger R_3 + \1^{\le N-2} \mathbb{K}_{\rm cr}^\dagger \mathbb{K}_{\rm cr} ) \left( f\Big(\frac{\N_+}{M} \Big) - f\Big(\frac{\N_+ +2}{M} \Big) \right) \\
&\quad + 2 \left( f\Big(\frac{\N_+}{M} \Big) - f\Big(\frac{\N_+ -2}{M} \Big) \right) (R_3 R_3^\dagger +\mathbb{K}_{\rm cr}\mathbb{K}_{\rm cr}^\dagger) \left( f\Big(\frac{\N_+}{M} \Big) - f\Big(\frac{\N_+ -2}{M} \Big) \right). \label{eq:b:l6}
\end{align}
Thanks to  (\ref{eq:Kcr*Kcr}), (\ref{eq:KcrKcr*}) and (\ref{eq:Rj*Rj-RjRj*}), all the terms $R_j^\dagger R_j$, $R_j R_j ^\dagger$, $\mathbb{K}_{\rm cr}^\dagger \mathbb{K}_{\rm cr}$ and $\mathbb{K}_{\rm cr}\mathbb{K}_{\rm cr}^\dagger$ can be bounded by
\begin{align} \label{eq:Rj*Rj-RjRj*-upper}
\frac{1}{N} \dGamma(QTQ)\N_+^2 + (\N_++1)^2.
\end{align}
The desired inequality then follows from the operator bound 
\begin{align} \label{eq:f-f}\left( f\Big( \frac{\N_+ \pm r}{M}\Big) -f\Big(\frac{\N_+}{M}\Big) \right)^2 \le \frac{r^2\|f'\|^2_{\infty}}{M^2}\1_{\F_+^{\le M+2}}
\end{align}
for $r=1,2.$
\end{proof}

We are now able to give the

\begin{proof}[Proof of Theorem \ref{thm:many-body-collective-excitation}] Recall the notation (\ref{eq:notH}), and let also 
$$\Phi_N := U_N \Psi_N \in \F_+^{\le N}.$$
Fix a Lipschitz function $f:\R\to [0,1]$ as in Lemma \ref{lem:localization} (we can take again $f(t)=\1_{(-\infty,1/2]} (t) + 2 (1-t)\1_{[1/2,1]}(t)$). For every $M\in [1,N]$, we denote 
$$ \Phi_N^{\le M} = f\Big(\frac{\N_+}{M} \Big) \Phi_N.$$

\noindent{\bf Step 1.} From (\ref{eq:N-body-kinetic}) and the Cauchy-Schwarz inequality, it follows that 
\begin{align} \label{eq:N-body-kinetic-1}
\bH_N ^2 \ge \frac{1}{4} (\dGamma(QTQ))^2 -C (\N_+ +1)^2.
\end{align}
Using (\ref{eq:N-body-kinetic-1}) and the equation 
$$(\bH_N -\lambda_N) \Phi_N = U_N (H_N - N\E_{\rm H}(u_0)) \Psi_N = 0\,,$$
we find that
$$
\lambda_N^2 = \| \bH_N \Phi_N \|^2 \ge \frac{1}{4} \| \dGamma(QTQ) \Phi_N \|^2 - C \| (\N_+ +1) \Phi_N\|^2\,.
$$
This implies the a-priori estimate
\begin{align} \label{eq:dGamma-PhiN}
\| \dGamma(QTQ) \Phi_N \|^2 \le C(\lambda_N^2+ 1+ N \langle \Phi_N, \N_+ \Phi_N \rangle).
\end{align}
Using again the equation $(\bH_N -\lambda_N) \Phi_N=0$ as well as Lemma \ref{le:N-body-commutator}, we obtain
\begin{align*}
& \| (\bH_N -\lambda_N)\Phi_N^{\le M} \|^2 = \left\| \left[ U_N H_N U_N, f\Big(\frac{\N_+}{M} \Big) \right] \Phi_N \right\|^2 \nn\\
&\le \frac{C}{M^2} \left\langle \Phi_N , \1^{\le M+2} \Big( N^{-1}\dGamma(QTQ)\N_+^2 + (\N_++1)^2 \Big) \Phi_N \right \rangle \nn\\
&\le \frac{1}{N^2} \| \dGamma(QTQ) \Phi_N \|^2 + \frac{C}{M^4} \| \1^{\le M+2} \N_+^2 \Phi_N \|^2 +  \frac{C}{M^2} \| \1^{\le M+2} (\N_++1) \Phi_N \|^2 .
\end{align*}
Using (\ref{eq:dGamma-PhiN}) and the operator inequality $\1^{\le M+2} \N_+ \le M+2$, we deduce that
\begin{align} \label{eq:loc-HN}
\| (\bH_N -\lambda_N)\Phi_N^{\le M} \|^2 \le \frac{C(\lambda_N^2 +1)}{N^2} + \frac{C}{M} \langle \Phi_N, (\N_+ +1) \Phi_N \rangle.
\end{align}
{\bf Step 2.} Using again (\ref{eq:N-body-kinetic-1}) and (\ref{eq:loc-HN}) we find that
\begin{align} \label{eq:dGamma-PhiN<=M}
\| \dGamma(QTQ) \Phi_N^{\le M} \|^2 &\le C \| \bH_N \Phi_N^{\le M}\|^2 + C \| (\N_+ +1) \Phi_N^{\le M} \|^2 \nn \\
&\le C (2 \| (\bH_N -\lambda_N) \Phi_N^{\le M}\|^2 + \lambda_N^2) + C M  \langle \Phi_N, (\N_+ +1) \Phi_N \rangle \nn\\
& \le C \lambda_N^2  + C M  \langle \Phi_N, (\N_+ +1) \Phi_N \rangle .
\end{align}
By Theorem \ref{thm:operator-bound} and (\ref{eq:dGamma-PhiN<=M}),
\begin{align} \label{eq:bHN-bH-PhiN<=M}
 \|(\bH_N -\1^{\le N}\bH)\Phi_N^{\le M}\|^2 &\le \frac{C}{N} \langle \Phi_N^{\le M}, (\dGamma(QTQ)\N_+^2+1 ) \Phi_N^{\le M} \rangle \nn\\
 &\le \frac{C M }{N} \| \dGamma(QTQ)  \Phi_N^{\le M} \| ^2 + \frac{C}{N M}  \| \N_+^2 \Phi_N^{\le M} \| ^2 + \frac{C}{N} \nn\\
 &\le \frac{CM(\lambda_N^2 +1)}{N} + \frac{CM^2}{N} \langle \Phi_N, \N_+ \Phi_N \rangle .
\end{align}
{\bf Step 3.} Combining (\ref{eq:loc-HN}) and (\ref{eq:bHN-bH-PhiN<=M}), we get
\begin{align*}
\|(\1^{\le N} \bH -\lambda_N)\Phi_N^{\le M}\|^2 &\le 2\|(\bH_N -\lambda_N)\Phi_N^{\le M}\|^2 + 2  \|(\bH_N -\1^{\le N}\bH)\Phi_N^{\le M}\|^2 \nn\\
&\le \frac{CM(\lambda_N^2 +1)}{N} + \left( \frac{C}{M}+\frac{CM^2}{N}\right)\langle \Phi_N, (\N_++1) \Phi_N \rangle.
\end{align*}
If $M=N^{1/3}\le N-2$, then the latter estimate implies that
\begin{align} \label{eq:bH-lambdaN-PhiN<=M}
\|(\bH -\lambda_N)\Phi_N^{\le M}\|^2 \le \frac{C\lambda_N^2 }{N^{2/3}} + \frac{C}{N^{1/3}}\langle \Phi_N, (\N_++1) \Phi_N \rangle.
\end{align}
Under the assumptions $|\lambda_N|\ll N^{1/3}$ and $\langle \Phi_N, \N_+ \Phi_N \rangle = \langle \Psi_N, \N_+ \Psi_N \rangle\ll N^{1/3}$, we can conclude from (\ref{eq:bH-lambdaN-PhiN<=M}) that
$$ \|(\bH -\lambda_N)\Phi_N^{\le M}\|^2 \to 0$$ 
as $N\to \infty$. Moreover, the choice $M=N^{1/3}$ also ensures that
$$
\|\Phi_N^{\le M}\|^2 \ge 1- \frac{\langle \Phi_N, \N_+ \Phi_N \rangle }{M} \to 1
$$
as $N\to \infty$. Thus we can take $\Phi_N'= \Phi_N^{\le M}/\|\Phi_N^{\le M}\|$ and complete the proof.
\end{proof}

\section{Proof of Theorem \ref{thm:excitation-spectrum}: Multiple condensations} \label{sec:multiple-Hartree}

In this section, we give the proof of Theorem \ref{thm:excitation-spectrum}. We shall use the min-max principle (see \cite{ReeSim4}), which we quickly recall below for the reader's convenience. If $A$ is a self-adjoint operator on a Hilbert space $\mathfrak{K}$ and $A$ is bounded from below, we can define the min-max values $\mu_1(A)\le \mu_2(A) \le \dots$
$$ \mu_j(A)=\inf_{\mathfrak{M}\subset \mathfrak{K}, \dim \mathfrak{M} =j} ~~~ \max_{\varphi\in \mathfrak{M}, \|\varphi\|=1 } \langle \varphi, A \varphi \rangle.$$
If $\lim_{j\to \infty}\mu_j(A)=\infty$, then $\{\mu_j(A)\}_{j=1}^\infty$ are all eigenvalues of $A$. In particular, under the condition that $T$ has compact resolvent, all eigenvalues of $H_N$ and $\bH_j$ are given by the min-max values. 

\begin{proof}[Proof of Theorem \ref{thm:excitation-spectrum}] 
{\em Upper bound.} We first prove the upper bound  
\bq \label{eq:mul-cond-upper-bound}\limsup_{N\to \infty} \Big( \mu_\ell(H_N) - Ne_{\rm H} \Big) \le \mu_\ell
\eq
for all $\ell \in \mathbb{N}$. For every $j\in \{1,2,\dots,J\}$, let $\mu_{j,1}\le \mu_{j,2} \le  \dots$ be the first eigenvalues of $\bH_j$ and let $\Phi_{j,1}, \Phi_{j,2}, \dots$ be the corresponding eigenvectors in $\F_{+j}$. Let $U_{N,j}: \gH^N \to \F_{+j}^{\le N}$ be the unitary transformation as in (\ref{eq:def-UN}) with $u_0$ replaced by $f_j$. Fix a Lipschitz function $f:\R\to [0,1]$ as in Lemma \ref{lem:localization} (we can take $f(t)=\1_{(-\infty,1/2]} (t) + 2(1-t)\1_{[1/2,1]}(t)$). Let $M=N^{1/3}$ and let
$$
\Psi_{N,j,k}:= U_{N,j}^\dagger f\Big(\frac{\N_{+j}}{M} \Big) \Phi_{j,k}.
$$
Then from Lemma \ref{lem:localization} one has
\begin{align} \label{eq:upper-bound-operator-cv} \lim_{N \to \infty} \Big\| \big( H_N -N e_{\rm H} - \mu_{j,k}\big) \Psi_{N,j,k} \Big\|^2 =0
\end{align}
for all $j\in \{1,2,\dots,J\}$ and $k\in \mathbb{N}$. Here we have used $\langle \Phi_{j,k}, \N_+ \Phi_{j,k} \rangle <\infty$, which follows from the non-degeneracy assumption (\ref{eq:BEC-fj-Hj}).

On the other hand, we will show that 
\begin{align} \label{eq:upper-bound-orthogonality}
\lim_{N\to \infty}\langle \Psi_{N,j,k}, \Psi_{N,j',k'}\rangle = \delta_{jj'} \delta_{kk'}
\end{align}
for all $j,j'\in \{1,2,\dots,J\}$ and $k,k'\in \mathbb{N}$. The upper bound (\ref{eq:mul-cond-upper-bound}) then follows from (\ref{eq:upper-bound-operator-cv}) and  (\ref{eq:upper-bound-orthogonality}) and the following consequence  of the min-max principle. 

\begin{lemma} Let $\{A_N\}_{N=1}^\infty$ be a sequence of self-adjoint operators on a Hilbert space. Assume that each $A_N$ is bounded from below and let $\mu_1(A_N),\mu_2(A_N)$ be the min-max values. Let $\lambda_1\le \lambda_2 \le \dots$ and let $\varphi_{N,1}, \varphi_{N,2},\dots\in D(A_N)$ such that
$$ \lim_{N\to \infty} (A_N - \lambda_k)\varphi_{N,k} =0\quad \text{and}\quad \lim_{N\to \infty} \langle \varphi_{N,k}, \varphi_{N,k'}\rangle =\delta_{kk'}$$
for all $k,k'\in \mathbb{N}$. Then for all $k\in \mathbb{N}$,
$$\limsup_{N\to \infty}\mu_k(A_N)\le \lambda_k.$$
\end{lemma}

The proof of this lemma is elementary and is left to the reader. It remains to verify the orthogonality (\ref{eq:upper-bound-orthogonality}). In fact, if $j= j'$, then using (\ref{eq:col-tradiction-orthogonal}) we have 
$$
\lim_{N\to \infty}\langle \Psi_{N,j,k}, \Psi_{N,j,k'}\rangle  = \lim_{N\to \infty} \left\langle f\Big(\frac{\N_{+j}}{M} \Big) \Phi_{j,k}, f\Big(\frac{\N_{+j}}{M} \Big) \Phi_{j,k'} \right\rangle_{\F_{+j}} = \delta_{kk'}
$$
for all $k,k'\in \mathbb{N}$. Now assume $j\ne j'$. A direct computation using (\ref{eq:alt-def-UN*}) and the fact that $f(t)=0$ when $t\ge 1$ shows that
\begin{align} \label{eq:formula-Psi-Njk}
\Psi_{N,j,k} &= \sum_{0\le \ell \le M}  f(\ell/M)  \frac{(a^\dagger(f_j))^{N-\ell}}{\sqrt{(N-\ell)!}}  \1_{\{\N_{+j}=\ell\}}  \Phi_{j,k}\\
& = \sum_{0\le \ell \le M} f(\ell/M)  \frac{1}{\sqrt{N! (N-\ell)! \ell! }} \sum_{\sigma \in S_N} \mathcal{U}_\sigma \left( f_j^{\otimes (N-\ell)}  \otimes  \1_{\{\N_{+j}=\ell\}}  \Phi_{j,k}  \right) \nn
\end{align}
where $S_N$ is the permutation group on $\{1,2,\dots,N\}$ and $\mathcal{U}_\sigma$ permutes the variables as 
\begin{align*}
& \left(\mathcal{U}_\sigma \left( f_j^{\otimes (N-\ell)}  \otimes  \1_{\{\N_{+j}=\ell\}} \Phi_{j,k}  \right) \right) (x_1,x_2,\dots,x_N) \\
&:= f_j(x_{\sigma(1)})\cdots f_j(x_{\sigma(N-\ell)}) \left( \1_{\{\N_+=\ell\}} \Phi_{j,k} \right)(x_{\sigma(N-\ell+1)},\dots,x_{\sigma(N)}).
\end{align*}
For all $j,j'\in \{1,2,\dots,J\}$, $k,k'\in \mathbb{N}$, $0\le \ell,\ell' \le M$ and $\sigma,\sigma'\in S_N$, an application of the simple bound $\|  \1_{\{\N_+=\ell\}} \Phi_{j,k} \|\leq 1$ 
yields
\begin{align*}
& \left| \left \langle \mathcal{U}_\sigma \left( f_j^{\otimes (N-\ell)}  \otimes \1_{\{\N_{+j}=\ell\}}  \Phi_{j,k} \right),  \,\mathcal{U}_{\sigma'} \left( f_{j'}^{\otimes (N-\ell')}  \otimes \1_{\{\N_{+j'}=\ell'\}} \Phi_{j',k'} \right)  \right\rangle \right| \\ & \le \left| \langle f_j, f_{j'} \rangle \right|^{N-2M}.
\end{align*}
From the latter bound,   the formula (\ref{eq:formula-Psi-Njk}) and $|f|\leq 1$, we obtain 
\begin{align*}
\left| \langle \Psi_{N,j,k}, \Psi_{N,j',k'} \rangle \right| &\le (M+1)^2 \frac{(N!)^2}{N!(N-M)!} \left| \langle f_j, f_{j'} \rangle \right|^{N-2M} \\
&\le (M+1)^2 N^M \left| \langle f_j, f_{j'} \rangle \right|^{N-2M} \to 0
\end{align*}
as $N\to \infty$ when $j\ne j'$. Here in the last convergence we have used the choice $M=N^{1/3}$ and the fact that $| \langle f_j, f_{j'} \rangle| <1$ when $j\ne j'$. Thus \eqref{eq:upper-bound-orthogonality} holds true and the upper bound (\ref{eq:mul-cond-upper-bound}) follows. 
\\\\
{\em Lower bound.} Now we prove the lower bound
\bq \label{eq:mul-cond-lower-bound}\liminf_{N\to \infty} \Big( \mu_\ell(H_N) - Ne_{\rm H} \Big) \ge \mu_\ell.
\eq
The proof is divided into several steps.
\\\\
{\bf Step 1} (Condensation). For every $\ell\in \mathbb{N}$, let $\mu_{\ell}(H_N)$ be the $\ell$-th eigenvalue of $H_N$ and let  $\Psi_{N,\ell}$ be the corresponding eigenvectors. For every $k\in \{1,2,\dots,N\}$ we can define the $k$-particle density matrix of $|\Psi_{N,\ell} \rangle \langle \Psi_{N,\ell}|$ by taking the partial trace over all but the $k$ first variables: 
$$
\gamma_{N,\ell}^{(k)}:= \Tr_{k+1\to N} |\Psi_{N,\ell} \rangle \langle \Psi_{N,\ell}|.
$$
Thus $\gamma_{N,\ell}^{(k)}$ is a non-negative trace class operator on $\gH^k$ with $\Tr \gamma_{N,\ell}^{(k)}=1$. Equivalently, we can define $\gamma_{N,\ell}^{(k)}$ from the formula (see e.g. \cite[Sec. 1.4]{Lewin-11})
\begin{align*}
\langle 0| a(f_{1}) \cdots a(f_k) & \gamma_{N,\ell}^{(k)} a^\dagger(g_1) \cdots a^\dagger(g_k) |0  \rangle \\
&= { N\choose k} ^{-1}\langle \Psi_{N,\ell} |   a^\dagger(g_k) \cdots a^\dagger(g_1) a(f_1)\cdots a(f_k) | \Psi_{N,\ell} \rangle
\end{align*}  
for all $f_1,\dots,f_k,g_1,\dots,g_k\in \gH$.

Using the assumption that $T$ has compact resolvent and the quantum de Finetti Theorem \cite{Stormer-69,HudMoo-75} as in \cite[Theorem 3.1]{LewNamRou-13}, for every $\ell\in\mathbb{N}$ we can find a subsequence of $\Psi_{N,\ell}$, still denoted by $\Psi_{N,\ell}$ for simplicity, and non-negative constants $\lambda_{\ell,j}$ with $\sum_{j=1}^J \lambda_{\ell,j}=1$ such that 
\bq \label{eq:deF}
\lim_{N\to \infty} \gamma_{N,\ell}^{(k)} = \sum_{j=1}^J \lambda_{\ell,j} |f_j^{\otimes k} \rangle \langle f_j^{\otimes k}|
\eq
in trace class for every $k\in \mathbb{N}$. Consequently, if we denote
$$\widehat n_j:=\frac{a^\dagger (f_j)a(f_j)}{N}\,,$$
then from (\ref{eq:deF}) and the fact that $\lim_{k\to \infty}\langle f_i^{\otimes k}, f_j^{\otimes k} \rangle = \delta_{ij}$, we have
\begin{align}\label{eq:ortho-deF}\lim_{k \to \infty} \lim_{N\to \infty}  \langle \widehat n_j ^k \Psi_{N,\ell}, \widehat n_{j'} ^k \Psi_{N,\ell} \rangle  = \lambda_{\ell,j} \delta_{jj'} . 
\end{align}
Consequently, 
\bq \label{eq:cv-X-0}
\lim_{k\to \infty} \lim_{N\to \infty} \left\| \Psi_{N,\ell}  - \sum_{j=1}^J \widehat n_j^k \Psi_{N,\ell}  \right\|=0
\eq
and
\bq \label{eq:BEC-niPsi}
\lim_{k\to \infty}  \lim_{N\to \infty}  \langle \widehat n_j^k \Psi_{N,\ell}, (\1- \widehat n_j ) \widehat n_j^k\Psi_{N,\ell} \rangle =0. 
\eq
When $\lambda_{\ell,j}\ne 0$, the convergence (\ref{eq:BEC-niPsi}) means that the vector $\widehat n_j^k \Psi_{N,\ell}/ \| \widehat n_j^k \Psi_{N,\ell}\| $ describes a  Bose-Einstein condensate in $f_j$ for large $N$ and suitable large $k$. 
\text{}\\\\
{\bf Step 2} (Splitting the energy).  Using the decomposition
$$
\Psi_{N,\ell}= \sum_{j=1}^{J} \widehat n_j^k \Psi_{N,\ell} + \Big( \1- \sum_{j=1}^{J} \widehat n_j^k \Big) \Psi_{N,\ell}
$$
we have
\begin{align} \label{eq:split-expansion}
0 &= \Big\langle \Psi_{N,\ell}, (H_N-\mu_\ell(H_N) ) \Psi_{N,\ell} \Big\rangle 
= \sum_{j=1}^{J}  \Big\langle  \widehat n_j^k \Psi_{N,\ell}, (H_N - \mu_\ell(H_N) )  \widehat n_j^k \Psi_{N,\ell} \Big\rangle \nn\\
&\quad + \Big\langle \Big( \1- \sum_{j=1}^{J} \widehat n_j^k \Big) \Psi_{N,\ell} , (H_N - \mu_\ell(H_N) ) \Big( \1- \sum_{j=1}^{J} \widehat n_j^k \Big) \Psi_{N,\ell} \Big\rangle \nn\\
&\quad +  \sum_{j=1}^J {\rm Re} \Big\langle  X_j , (H_N - \mu_\ell(H_N) )  \widehat n_j^k \Psi_{N,\ell} \Big\rangle
\end{align}
where
$$
X_j := \sum_{i,i\ne j} \widehat n_i^k  \Psi_{N,\ell} + 2  \Big( \1- \sum_{i=1}^{J} \widehat n_i^k \Big) \Psi_{N,\ell} = \Big( 2 \1 - \widehat n_j^k - \sum_{i=1}^J \widehat n_i^k \Big) \Psi_{N,\ell}.
$$

Let us estimate the last sum on the right side of (\ref{eq:split-expansion}). From the equation 
\begin{equation}\label{eve}
(H_N - \mu_\ell(H_N) ) \Psi_{N,\ell} =0
\end{equation}
we have
\begin{align} \label{eq:HN-mu-ell-commutator}
(H_N - \mu_\ell(H_N) ) \widehat n_j^k \Psi_{N,\ell} = [H_N, \widehat n_j^k]\Psi_{N,\ell}.
\end{align}
The commutator $[H_N, \widehat n_j^k]$ can be estimated by following the proof of Lemma \ref{le:N-body-commutator}, with the aid of the argument in Remark \ref{re:HN-bHN-commutator}. Since $\widehat n_j = 1 - \mathcal{N}_{+j}/N$, the operator $\widehat n_j^k$ can be viewed as a function of $\mathcal{N}_{+j}$, and the same bound (\ref{eq:b:l6}) as in the proof of Lemma~\ref{le:N-body-commutator} applies. Moreover, the upper bound in (\ref{eq:Rj*Rj-RjRj*-upper}) can simply be replaced by $C(N \dGamma(Q_jTQ_j)+N^2)$, here $Q_j=1-|f_j \rangle \langle f_j|$. Instead of (\ref{eq:f-f}) we can bound 
\begin{align*}
\left(\frac{(a(f_j)^\dagger a(f_j) \pm r)^k -(a(f_j)^\dagger a(f_j))^k }{N^k} \right)^2 
&\le \left( \frac{k r}{N} \widehat n_j ^{k-1} + \frac{k^2 r^2}{N^2} \left(\widehat n_j + \frac r N\right)^{k-2} \right)^2 \\
&\le \frac{2k^2 r^2}{N^2} \widehat n_j ^{2k-2} + \frac{2 k^4 r^4}{N^4}\left( 1+ \frac 2  N \right)^{2k -4}
\end{align*}
for $r=1,2$, which follows from
$$
 |(x \pm y)^k - x^k|  \leq k y(x+y)^{k-1} \leq k y x^{k-1} + k^2 y^2 (x+y)^{k-2}
$$
for $x \ge 0$ and $y \ge 0$.
We then obtain
\begin{align}\label{eq:commu-HN-Ni}
-\big[H_N, \widehat n_j^k \big]^2 \le C \left( \frac{ \dGamma(Q_jTQ_j)}{N}+ 1 \right) \left( k^2 \widehat  n_j^{2k-2} +\frac{k^4}{N^2} \right) 
\end{align} 
for $N \geq k$. It follows  from (\ref{eq:HN-mu-ell-commutator}) and (\ref{eq:commu-HN-Ni}) that
\begin{align}& \left| \Big\langle   X_j , (H_N - \mu_\ell(H_N) )  \widehat n_j^k \Psi_{N,\ell} \Big\rangle \right|^2 \nn\\
&= \left| \Big\langle   X_j , \big[H_N, \widehat n_j^k \big]  \Psi_{N,\ell} \Big\rangle \right|^2 \le \left\| \big[H_N, \widehat n_j^k \big] X_j \right\|^2 \nn\\
&\le C \left\langle X_j, \left( \frac{ \dGamma(Q_jTQ_j)}{N}+ 1 \right) \left( k^2 \widehat  n_j^{2k-2} +\frac{k^4}{N^2} \right) X_j  \right \rangle \nn \\
& \le C \left\| \left( \frac{ \dGamma(Q_jTQ_j)}{N}+ 1 \right) X_j \right\| \, \left\| \left( k^2 \widehat  n_j^{2k-2} +\frac{k^4}{N^2} \right) X_j \right\|. \label{eq:X-ni-estimate-0}
\end{align}

In order to estimate the right side of (\ref{eq:X-ni-estimate-0}), let us prove the simple bound \begin{align} \label{eq:DT2<=N2} \|\dGamma(Q_i T Q_i) \widehat n_j^k \Psi_{N,\ell}  \|^2 \le C  \left(  |\mu_\ell(H_N)|^2 + N^2 \right)
\end{align}
for all $0\le k \le N$, $\ell\in \mathbb{N}$ and $i,j\in \{1,2,\dots,J\}$. First, when $k=0$, we can use (\ref{eq:N-body-kinetic-1}) and the eigenvalue equation $H_N \Psi_{N,\ell} = \mu_\ell(H_N) \Psi_{N,\ell}$ to bound
\begin{align} \label{eq:kinetic-HN-k=0}
\Big\|  \dGamma( Q_i T Q_i)  \Psi_{N,\ell}  \Big\|^2  \le 4 |\mu_\ell(H_N)|^2 + C N^2
\end{align}
for all $i\in \{1,2,\dots,J\}$. When $k\ge 1$, on the other hand, we can bound
\begin{align*} 
& \Big\|  \dGamma(Q_i  T Q_i) \widehat n_j^k \Psi_{N,\ell}  \Big\|^2 \le 4\Big\| H_N  \widehat n_j^k \Psi_{N,\ell} \Big\|^2 + C  N^2 \nn\\
&  \le 8\Big\| (H_N - \mu_\ell(H_N) ) \widehat n_j^k \Psi_{N,\ell} \Big\|^2 + 8 |\mu_\ell(N)|^2 +  C N^2 \nn\\
& = 8 \Big\| [H_N, \widehat n_j^k]\Psi_{N,\ell} \Big\|^2 + 8 |\mu_\ell(N)|^2 +  C N^2 \nn\\
& \le C \left\langle \Psi_{N,\ell}, \left( \frac{ \dGamma(Q_jTQ_j)}{N}+ 1 \right) \left( k^2 \widehat  n_j^{2k-2} +\frac{k^4}{N^2} \right) \Psi_{N,\ell} \right\rangle +  8 |\mu_\ell(N)|^2 +  C N^2 \,,
\end{align*}
where we used (\ref{eq:HN-mu-ell-commutator}) and (\ref{eq:commu-HN-Ni}) in the last step. By combining the latter estimate with $\widehat n_j \le \1$ and (\ref{eq:kinetic-HN-k=0}), using the assumption $k\leq N$, we obtain (\ref{eq:DT2<=N2}).

Finally, we turn back to the right side of (\ref{eq:X-ni-estimate-0}). Using (\ref{eq:DT2<=N2}) and the triangle inequality, we get
\begin{equation}\label{eq:X-ni-estimate-1}
\left\| \left( \frac{ \dGamma(Q_jT Q_j)}{N}+ 1 \right) X_j \right\| \le C (J+1) \left( 1 + \frac{|\mu_\ell(H_N)|}N \right) \,.
\end{equation}
Note that for any fixed $\ell$, $|\mu_\ell(H_N)|/N$ is uniformly bounded in $N$. 
Moreover, from (\ref{eq:deF}) and the fact that $\langle f_i^{\otimes k}, f_j^{\otimes k} \rangle$ decays exponentially in $k$ when $i\ne j$, it follows that
$$
\lim_{k \to \infty} \lim_{N\to \infty} \left\| k^2 \widehat  n_j^{2k-2} \widehat  n_i^{k} \Psi_{N,\ell}\right\|=0
$$
for all $i\ne j$, and hence
\begin{align} \label{eq:X-ni-estimate-2}
\lim_{k \to \infty} \lim_{N\to \infty} \left\| \left( k^2 \widehat  n_j^{2k-2} +\frac{k^4}{N^2} \right) X_j \right\|=0.
\end{align}
From (\ref{eq:X-ni-estimate-0}), (\ref{eq:X-ni-estimate-1})  and (\ref{eq:X-ni-estimate-2}), we conclude that  
\begin{align} \label{eq:split-cross-term} \lim_{k \to \infty} \lim_{N\to \infty} \left| \Big\langle  X_j , (H_N - \mu_\ell(H_N) )  \widehat n_j^k \Psi_{N,\ell} \Big\rangle \right| =0 
\end{align}
for all $j\in \{1,2,\dots,J\}$. As a consequence of (\ref{eq:split-cross-term}) and (\ref{eq:split-expansion}), we have 
\begin{align} \label{eq:mulH-split}
 \lim_{k \to \infty} \lim_{N\to \infty} & \left( \sum_{j=1}^{J}  \Big\langle   \widehat n_j^k \Psi_{N,\ell}, (H_N - \mu_\ell(H_N) )  \widehat n_j^k \Psi_{N,\ell} \Big\rangle  \right.\\
&+  \left. \Big\langle \Big( \1- \sum_{j=1}^{J} \widehat n_j^k \Big) \Psi_{N,\ell} , (H_N - \mu_\ell(H_N) ) \Big( \1- \sum_{j=1}^{J} \widehat n_j^k \Big) \Psi_{N,\ell} \Big\rangle\right) = 0 . \nn
 \end{align}
{\bf Step 3} (Localization). Now we consider each term
$\big\langle   \widehat n_j^k \Psi_{N,\ell}, H_N \widehat n_j^k \Psi_{N,\ell} \big\rangle $
separately. Since $\widehat n_j^k \Psi_{N,\ell}$ satisfies the complete condensation in $f_j$, in the sense of \eqref{eq:BEC-niPsi}, it is reasonable to use the unitary transformation $U_{N,j}: \gH^N \to \F_{+j}^{\le N}$ which is defined as in (\ref{eq:def-UN}) with $u_0$ replaced by $f_j$. We have
$$
\Big\langle   \widehat n_j^k \Psi_{N,\ell}, (H_N - Ne_{\rm H} )  \widehat n_j^k \Psi_{N,\ell} \Big\rangle  =  \Big\langle   U_{N,j} \widehat n_j^k \Psi_{N,\ell},  \bH_{N,j} U_{N,j} \widehat n_j^k \Psi_{N,\ell} \Big\rangle$$
where
$$
\bH_{N,j}:= U_{N,j}(H_N -N e_{\rm H}) U_{N,j}^\dagger.
$$
By Theorem \ref{thm:operator-bound}, we have the quadratic form bound 
\begin{align} \label{eq:bHNj-bHj-quadratic}
\bH_{N,j} - \1_{\F_{+j}^{\le N}} \bH_j \1_{\F_{+j}^{\le N}} &\ge -\frac{C}{\sqrt{N}} (\dGamma(Q_j T Q_j))^{1/2} \N_{+j} - \frac{C}{\sqrt{N}} \nn\\
& \ge -C{\sqrt{\frac{\N_{+j}}{N}}} \dGamma(Q_j T Q_j) - \frac{C}{\sqrt{N}} 
\end{align}
on $\F_{+j}^{\le N}$, where again $Q_j=\1-|f_j \rangle \langle f_j|$ and $\N_{+j}=\dGamma(Q_j)$. Heuristically, the kinetic term $ \dGamma(Q_j T Q_j) $ is of the same order as $\bH_j$. In fact, the non-degeneracy condition \eqref{eq:BEC-fj-Hj} implies that (see \cite[Theorem 1]{LewNamSerSol-13})
\begin{align} \label{eq:non-degeneracy-dGamma<=bH}
\dGamma(Q_j TQ_j)\le C(\bH_j+C)\,.
\end{align}  
Therefore, the bound (\ref{eq:bHNj-bHj-quadratic}) is useful in the region $\N_{+j} \ll N$. To proceed, we shall follow the localizing strategy in \cite{LewNamSerSol-13}. Let $f,g:\R\to [0,1]$ be Lipschitz functions such that
$$f^2+g^2=1\quad \text{and}\quad g\1_{(-\infty,1/2]}  =0 = f\1_{(1,\infty]}.$$
For any $M\in [1,N-2]$, an IMS-type estimate (see \cite[Lemma 24]{LewNamSerSol-13}) yields
\begin{align*}
\bH_{N,j} &\ge f\Big( \frac{\N_{+j}}{M}\Big) \bH_{N,j}  f\Big( \frac{\N_{+j}}{M}\Big) + g\Big( \frac{\N_{+j}}{M}\Big) \bH_{N,j}  g\Big( \frac{\N_{+j}}{M}\Big)\\
& \quad - \frac{C}{M^2} \left( \dGamma( Q_j TQ_j) +N \right).
\end{align*}
On the other hand, from (\ref{eq:bHNj-bHj-quadratic}) and (\ref{eq:non-degeneracy-dGamma<=bH}) we find that  
$$
f\Big( \frac{\N_{+j}}{M}\Big) \bH_{N,j} f\Big( \frac{\N_{+j}}{M}\Big) \ge \left( 1-C\sqrt{\frac{M}{N}}  \right) f\Big( \frac{\N_{+j}}{M}\Big) \bH_{j}f\Big( \frac{\N_{+j}}{M}\Big) - C\sqrt{\frac{M}{N}}.
$$
Thus
\begin{align} \label{eq:local-bHNj-bHj}
\bH_{N,j} &\ge \left( 1-C\sqrt{\frac{M}{N}}  \right) f\Big( \frac{\N_{+j}}{M}\Big) \bH_{j} f\Big( \frac{\N_{+j}}{M}\Big) +  g\Big( \frac{\N_{+j}}{M}\Big) \bH_{N,j}  g\Big( \frac{\N_{+j}}{M}\Big)\nn\\
&\quad - \frac{C}{M^2} \left( \dGamma( Q_j TQ_j) +N \right) - C\sqrt{\frac{M}{N}} .
\end{align}
Now we take the expectation of the quadratic inequality (\ref{eq:local-bHNj-bHj}) against $U_{N,j}\widehat n_j^k \Psi_{N,\ell}$. Note that
$$
g\Big( \frac{\N_{+j}}{M}\Big) \bH_{N,j}  g\Big( \frac{\N_{+j}}{M}\Big) \ge \mu_1(\bH_{N,j}) g^2\Big( \frac{\N_{+j}}{M}\Big)
$$  
and
\begin{align*} & \langle U_{N,j} \widehat n_j^k \Psi_{N,\ell}, \dGamma(Q_j T Q_j) U_{N,j} \widehat n_j^k \Psi_{N,\ell} \rangle \\
&=  \langle  \widehat n_j^k \Psi_{N,\ell}, \dGamma(Q_j T Q_j) \widehat n_j^k \Psi_{N,\ell} \rangle \le C  \left(  k \sqrt{\ell}+ \frac{  3^k \sqrt{\ell}}{N}+ \ell N \right)
\end{align*}
due to  inequality (\ref{eq:DT2<=N2}). Thus from (\ref{eq:local-bHNj-bHj}) we obtain
\begin{align} \label{eq:local-bHNj-bHj-1}
&\Big\langle   U_{N,j} \widehat n_j^k \Psi_{N,\ell}, \bH_{N,j}   U_{N,j} \widehat n_j^k \Psi_{N,\ell} \Big\rangle - \mu_1(\bH_{N,j}) \left\| g\Big( \frac{\N_{+j}}{M}\Big) U_{N,j}\widehat n_j^k \Psi_{N,\ell} \right\|^2 \nn\\
&\ge \left( 1-C\sqrt{\frac{M}{N}}  \right) \left\langle  f\Big( \frac{\N_{+j}}{M}\Big) U_{N,j}  \widehat n_j^k \Psi_{N,\ell}, \bH_j    f\Big( \frac{\N_{+j}}{M}\Big) U_{N,j}\widehat n_j^k \Psi_{N,\ell} \right\rangle \nn\\
&\quad -\frac{C}{M^2}  \left(  k \sqrt{\ell}+ \frac{  3^k \sqrt{\ell}}{N}+ \ell N \right)   - C\sqrt{\frac{M}{N}}.
\end{align}
{\bf Step 4} (Ground state energy) Since $H_N - \mu_1(H_N) \ge 0$, we can deduce from (\ref{eq:mulH-split}) that for all $j\in \{1,2,\dots,J\}$,
\begin{align*} 
 \limsup_{k\to \infty} \limsup_{N\to \infty} \langle   \widehat n_j^k \Psi_{N,1}, (H_N - \mu_1(H_N) )  \widehat n_j^k \Psi_{N,1} \rangle\le 0.
\end{align*}
Since $\mu_1(\bH_{N,j})=\mu_1(H_N)-Ne_{\rm H}$, the latter inequality is equivalent to
\begin{align} \label{eq:GSE-split} 
 \limsup_{k\to \infty} \limsup_{N\to \infty} \left\langle   U_{N,j} \widehat n_j^k \Psi_{N,1}, (\bH_{N,j} - \mu_1(\bH_{N,j})   U_{N,j} \widehat n_j^k \Psi_{N,1} \right\rangle   \le 0.
\end{align}

Now we choose $M=N/k^2$ and estimate both sides of (\ref{eq:local-bHNj-bHj-1}). Since 
$$
g^2\Big( \frac{\N_{+j}}{M}\Big) \le \frac{2\N_{+j}}{M} = 2k^2 \left( \1- \widehat n_j\right) 
$$
and the convergence in (\ref{eq:BEC-niPsi}) is exponentially fast in $k$, we find that
\bq \label{eq:g-vanish}
\lim_{k\to \infty}\lim_{N\to \infty} \left\| g\Big( \frac{\N_{+j}}{M}\Big) U_{N,j}\widehat n_j^k \Psi_{N,1} \right\| =0,
\eq
for all $j\in \{1,2,\dots,J\}$. Since $g^2+f^2=1$, we get
\begin{align} \label{eq:1-f-lambda}
\lim_{k\to \infty}\lim_{N\to \infty} \left\| f\Big( \frac{\N_{+j}}{M}\Big) U_{N,j}\widehat n_j^k \Psi_{N,\ell} \right\|^2 = \lim_{k\to \infty}\lim_{N\to \infty} \left\| \widehat n_j^k \Psi_{N,\ell} \right\|^2 =  \lambda_{j,\ell} 
\end{align}
for all $j\in \{1,2,\dots,J\}$. Using (\ref{eq:GSE-split}) and (\ref{eq:1-f-lambda}) we can estimate the left side of (\ref{eq:local-bHNj-bHj-1}) as
\begin{align*}
& \liminf_{k\to \infty} \liminf_{N\to \infty} \left( \Big\langle   U_{N,j} \widehat n_j^k \Psi_{N,1}, \bH_{N,j}   U_{N,j} \widehat n_j^k \Psi_{N,1} \Big\rangle \right. \nn\\
&\quad \quad\quad\quad \quad\quad\quad\quad\quad\left.- \mu_1(\bH_{N,j}) \Big\| g\Big( \frac{\N_{+j}}{M}\Big) U_{N,j}\widehat n_j^k \Psi_{N,1} \Big\|^2 \right)\nn\\
&= \liminf_{k\to \infty} \liminf_{N\to \infty} \left( \Big\langle   U_{N,j} \widehat n_j^k \Psi_{N,1}, (\bH_{N,j}-\mu_1(\bH_{N,j})   U_{N,j} \widehat n_j^k \Psi_{N,1} \Big\rangle \right. \nn\\
&\quad \quad\quad\quad \quad\quad\quad\quad\quad\left. + \mu_1(\bH_{N,j}) \Big\| f\Big( \frac{\N_{+j}}{M}\Big) U_{N,j}\widehat n_j^k \Psi_{N,1} \Big\|^2 \right)\nn\\
&\le \lambda_{j,1} \liminf_{N\to \infty} \mu_1(\bH_{N,j}) = \lambda_{j,1} \liminf_{N\to \infty}( \mu_1(H_N)-Ne_{\rm H})
\end{align*}
for all $j\in \{1,2,\dots,J\}$. Then taking the same limit on the right side of (\ref{eq:local-bHNj-bHj-1}) and using the simple bound $\bH_j\ge \mu_1(\bH_j)$, we find that
\begin{align}  \label{eq:GSE-lambda-1}
&\quad \lambda_{j,1} \liminf_{N\to \infty} ( \mu_1(H_N)-Ne_{\rm H}) \nn\\
&\ge \liminf_{k\to \infty} \liminf_{N\to \infty} \left\langle  f\Big( \frac{\N_{+j}}{M}\Big) U_{N,j}  \widehat n_j^k \Psi_{N,1}, \bH_j  f\Big( \frac{\N_{+j}}{M}\Big) U_{N,j}\widehat n_j^k \Psi_{N,1} \right\rangle \nn\\
&\ge \lambda_{j,1} \mu_1(\bH_j) \ge \lambda_{j,1} \min_{1\le i \le J}\mu_1(\bH_i)
\end{align}
for all $j\in \{1,2,\dots,J \}$. Taking the sum over $j$ and using $\sum_{j=1}^J  \lambda_{j,1} =1$ we obtain the lower bound
$$
\liminf_{N\to \infty} ( \mu_1(H_N)-Ne_{\rm H}) \ge \min_{1\le j\le J}  \mu_1(\bH_j) .
$$
Together with the upper bound \eqref{eq:mul-cond-upper-bound}, we then conclude that
\begin{align} \label{eq:final-cv-mu1}
\lim_{N\to \infty} ( \mu_1(H_N)-Ne_{\rm H}) = \min_{1\le j\le J}  \mu_1(\bH_j) .
\end{align}

From the above proof, we can also deduce easily the structure of the ground state $\Psi_{N,1}$. The convergence (\ref{eq:final-cv-mu1}) implies that we always have  equality in (\ref{eq:GSE-lambda-1}) for all $j\in \{1,2,\dots,J\}$. Consequently, for all $j\in \{1,2,\dots,J\}$ we have
\bq \label{eq:lambda-j-1=0}
\lambda_{j,1}=0\quad \text{if}\quad \mu_1(\bH_j)> \min_{1\le i \le J} \mu_1(\bH_i)
\eq
and  
$$
\lim_{k\to \infty} \lim_{N\to \infty} \left\langle  f\Big( \frac{\N_{+j}}{M}\Big) U_{N,j}  \widehat n_j^k \Psi_{N,1}, (\bH_j - \mu_1(\bH_j) )  f\Big( \frac{\N_{+j}}{M}\Big) U_{N,j}\widehat n_j^k \Psi_{N,1} \right\rangle
= 0.$$
The non-degeneracy condition (\ref{eq:BEC-fj-Hj}) implies that $\mu_1(\bH_j)<\mu_2(\bH_j)$ (see \cite[Theorem 1]{LewNamSerSol-13}), and hence we can deduce from the latter convergence that
$$
\lim_{k\to \infty} \lim_{N\to \infty} f\Big( \frac{\N_{+j}}{M}\Big) U_{N,j}  \widehat n_j^k \Psi_{N,1} = \sqrt{\lambda_{j,1}} \Phi_{j,1},
$$
for all $j\in \{1,2,\dots,J\}$, where $\Phi_{j,1}$ is the unique ground state of $\bH_j$ (up to a complex phase). Because of (\ref{eq:g-vanish}), the latter convergence is equivalent to  
$$
\lim_{k\to \infty} \lim_{N\to \infty} U_{N,j}  \widehat n_j^k \Psi_{N,1} = \sqrt{\lambda_{j,1}}  \Phi_{j,1}.
$$
In combination with (\ref{eq:cv-X-0}), we can conclude that
\begin{align} \label{eq:final-cv-N-body-ground-state} \lim_{N\to \infty} \left\| \Psi_{N,1} - \sum_{j=1}^J \sqrt{\lambda_{j,1}} U_{N,j}^\dagger\Phi_{j,1} \right\| = 0\,,
\end{align}
where $U_{N,j}^\dagger$ is extended by $0$ outside $\F_{+j}^{\le N}$, i.e. $U_{N,j}^\dagger\Phi_{j,1} :=U_{N,j}^\dagger \1_{\F_{+,j}^{\le N}}\Phi_{j,1}$.
\text{}\\\\
{\bf Step 5} (Higher eigenvalues) For every $\ell\in \mathbb{N}$, from the upper bound on $\mu_\ell(H_N)$ in (\ref{eq:mul-cond-upper-bound}) and the convergence of $\mu_1(H_N)$ in (\ref{eq:final-cv-mu1}), it follows that $H_N-\mu_\ell(H_N)$ is bounded below by a constant independent of $N$. Therefore, using (\ref{eq:cv-X-0}) we can again remove the last term of (\ref{eq:mulH-split}) and obtain 
\begin{align*}  
 \limsup_{k\to \infty} \limsup_{N\to \infty}\sum_{j=1}^{J}  \langle   \widehat n_j^k \Psi_{N,\ell}, (H_N - \mu_\ell(H_N) )  \widehat n_j^k \Psi_{N,\ell} \rangle \le 0\,.
\end{align*}
This can be rewritten as
\begin{align} \label{eq:mu-ell-HN-expectation}
&\liminf_{N\to \infty} (\mu_\ell(H_N) - N e_{\rm H}) \ge \limsup_{k\to \infty} \limsup_{N\to \infty} \sum_{j=1}^{J} \langle   \widehat n_j^k \Psi_{N,\ell}, (H_N - N e_{\rm H}) \widehat n_j^k \Psi_{N,\ell} \rangle \nn\\
&\quad \quad \quad\quad = \limsup_{k\to \infty} \limsup_{N\to \infty} \sum_{j=1}^{J} \Big\langle   U_{N,j} \widehat n_j^k \Psi_{N,\ell}, \bH_{N,j}   U_{N,j} \widehat n_j^k \Psi_{N,\ell} \Big\rangle. 
\end{align}
Using (\ref{eq:mu-ell-HN-expectation}) and (\ref{eq:local-bHNj-bHj-1})  (with the same choice $M=N/k^2$), we find the following analogue of (\ref{eq:GSE-lambda-1}) 
\begin{align}  \label{eq:GSE-lambda-ell}
&\quad \liminf_{N\to \infty} ( \mu_\ell(H_N)-Ne_{\rm H}) \nn\\
&\ge \liminf_{k\to \infty} \liminf_{N\to \infty} \sum_{j=1}^J \left\langle  f\Big( \frac{\N_{+j}}{M}\Big) U_{N,j}  \widehat n_j^k \Psi_{N,\ell}, \bH_j  f\Big( \frac{\N_{+j}}{M}\Big) U_{N,j}\widehat n_j^k \Psi_{N,\ell} \right\rangle 
\end{align}
for all $j\in \{1,2,\dots,J \}$ and $\ell\in \mathbb{N}$. We will estimate the right side of (\ref{eq:GSE-lambda-ell}) using the following min-max principle, whose proof is elementary and is left to the reader.

\begin{lemma}\label{le:min-max-lower} Let $A$ be a self-adjoint operator on a Hilbert space. Assume that $A$ is bounded from below and all min-max values $\mu_1(A), \mu_2(A), \dots$ are eigenvalues with the corresponding eigenvectors $\Phi_1,\Phi_2, \dots$. Let $\{\varphi_N\}_{N=1}^\infty$ be a sequence of normalized vectors satisfying $\langle \Phi_n, \varphi_N  \rangle \to 0$ as $N\to \infty$ for all $n\in \{1,2,\dots,L\}$. Then
$$
\liminf_{N\to \infty} \langle \varphi_N, A \varphi_N \rangle \ge \mu_{L+1}(A).
$$
Moreover, if 
$$
\lim_{N\to \infty} \langle \varphi_N, A \varphi_N \rangle = \mu_{L+1}(A)
$$
and $\mu_L(A)<\mu_{L+1}(A)=\mu_{L'}(A)<\mu_{L'+1}(A)$, then there is a subsequence of $\varphi_N$ (still denoted by $\varphi_N$ for short) such that
$$
\lim_{N\to \infty} \Big\| \varphi_N - \sum_{\ell=L+1}^{L'} \theta_{\ell} \Phi_{\ell} \Big\| =0
$$ 
for some complex numbers $\theta_{\ell}$ satisfying $\sum_{\ell=L+1}^{L'} |\theta_\ell|^2=1$.
\end{lemma}

Recall that $\{\mu_\ell \}_{\ell=1}^\infty$ denotes the increasing sequence which is rearranged from the union (counting multiplicity) of the eigenvalues of the $\bH_j$'s. Let $L\in \mathbb{N}$ such that $\mu_L<\mu_{L+1}$. We assume that the $L$ numbers $\mu_1,\dots,\mu_L$ consist of $r_j$ eigenvalues of $\bH_j$ (counting multiplicity) with the corresponding eigenvectors $\{\Phi_{j,i}\}_{i=1}^{r_j}$, for all $j\in \{1,2,\dots,J\}$. Thus $r_j \ge 0$ and $\sum_{j=1}^L r_j=L$. We shall show that
\bq \label{eq:cv-mu-ell-induction}
 \lim_{N\to \infty}(\lambda_\ell(H_N) - Ne_{\rm H}) = \mu_\ell
 \eq
for all $\ell \in \{1,2,\dots,L\}$. Moreover, we shall also show that for the corresponding eigenfunctions $\{\Psi_{N,\ell}\}_{\ell=1}^L$ of $H_N$ 
there is a subsequence  (still denoted by $\Psi_{N,\ell}$ for short)  satisfying 
\begin{align} \label{eq:Psi-N-ell-induction}
&\big( \Psi_{N,1},\dots,\Psi_{N,L} \big)^T \\
&\quad \quad = \mathfrak{A}_L \big(U_{N,1}^\dagger\Phi_{1,1},\dots, U_{N,1}^\dagger\Phi_{1,n_1},\dots, U_{N,J}^\dagger\Phi_{J,1},\dots,U_{N,J}^\dagger\Phi_{J,r_J} \big)^T + \mathfrak{R}_{N,L} \nn
\end{align}
where $\mathfrak{A}_L$ is a $L\times L$ complex matrix independent of $N$, and $\|\mathfrak{R}_{N,L}\|_{(\gH^N)^L}\to 0$ as $N\to \infty$. Note that due to the orthonormality $\langle \Psi_{N,\ell}, \Psi_{N,\ell'}\rangle=\delta_{\ell \ell'}$ and the fact that $\lim_{N\to \infty} \langle U_{N,j}^\dagger \Phi_{j,i}, U_{N,j'}^\dagger \Phi_{j',i'} \rangle = \delta_{jj'} \delta_{ii'}$ because of (\ref{eq:upper-bound-orthogonality}), the matrix $\mathfrak{A}_L$ satisfying (\ref{eq:Psi-N-ell-induction}) must necessarily be unitary. 

We shall prove (\ref{eq:cv-mu-ell-induction}) and (\ref{eq:Psi-N-ell-induction}) using an induction argument. 
\\\\
{\it Base case}. First, we take $L$ such that $\mu_1=\mu_L<\mu_{L+1}$. Since $\mu_1(\bH_j)<\mu_2(\bH_j)$ for all $j\in \{1,2,\dots,J\}$, it follows that $L\le J$ and there are exactly $L$ numbers $i(1),i(2),\dots,i(L)\in \{1,2,\dots,J\}$ such that 
$$ \mu_1(\bH_{i(1)})=\dots=\mu_1(\bH_{i(L)})=\mu_1 < \mu_1(\bH_j)$$
for all $j\notin \{i(1),\dots,i(L)\}$. For every $\ell\in \{1,2,\dots,L\}$, from the upper bound on $\mu_\ell(H_N)$ in (\ref{eq:mul-cond-upper-bound}) and the convergence of $\mu_1(H_N)$ in (\ref{eq:final-cv-mu1}) we have
$$
\limsup_{N\to \infty} (\mu_\ell (H_N)- \mu_1(H_N)) \le \mu_\ell - \mu_1  =0.
$$
and hence (\ref{eq:cv-mu-ell-induction}) holds true for $\ell\in \{1,2,\dots,L\}$. 

Moreover, from (\ref{eq:GSE-lambda-ell}), by using the same argument applied to the ground state $\Psi_{N,1}$ in (\ref{eq:final-cv-N-body-ground-state}), we have $\lambda_{j,\ell}=0$ if $j\notin \{i(1),\dots,i(L)\}$ and we can find complex numbers $\theta_{j,\ell} \in \C$ such that $|\theta_{j,\ell}|=\sqrt{\lambda_{j,\ell}}$ and 
\begin{align} \label{eq:cv-L-first-ev} \lim_{N\to \infty} \left\| \Psi_{N,\ell} - \sum_{j=1}^{L}  \theta_{i(j),\ell} U_{N,i(j)}^\dagger\Phi_{i(j),1} \right\| = 0
\end{align}
for $\ell\in \{1,2,\dots,L\}$, where $\Phi_{j,1}$ is the unique ground state of $\bH_j$. %(when $\ell=1$, from \eqref{eq:final-cv-N-body-ground-state} it follows that $\theta_{j,1}=\sqrt{\lambda_{j,1}}$ for all $j=\{1,2,...,J\}$). 
We can rewrite (\ref{eq:cv-L-first-ev}) as 
$$
\left( \begin{gathered}
   \Psi_{N,1} \\
   \Psi_{N,2} \\
   \vdots  \\
   \Psi_{N,L} \\ 
\end{gathered}  \right) = \left( {\begin{array}{*{20}{c}}
   \theta_{i(1),1} & \ldots  & \theta_{i(L),1} \\ 
   \theta_{i(1),2} & \ldots  & \theta_{i(L),2} \\ 
 \vdots & & \vdots \\
  \theta_{i(1),L} & \ldots  & \theta_{i(L),L} \\ 
\end{array}} \right) 
\left( \begin{gathered}
   U_{N,i(1)}^\dagger\Phi_{i(1),1}  \\
   U_{N,i(2)}^\dagger\Phi_{i(2),1}  \\
   \vdots  \\
   U_{N,i(L)}^\dagger\Phi_{i(L),1} \\ 
\end{gathered}  \right) + \mathfrak{R}_{N,L}
$$
where $\|\mathfrak{R}_{N,L}\|_{(\gH^N)^L}\to 0$ as $N\to \infty$. Thus (\ref{eq:Psi-N-ell-induction}) holds true.
\\\\
{\it Inductive step.} Now let $L,L'\in \mathbb{N}$ be arbitrary indexes such that $\mu_L<\mu_{L+1}=\mu_{L'}<\mu_{L'+1}$. We shall prove that if (\ref{eq:cv-mu-ell-induction}) and (\ref{eq:Psi-N-ell-induction}) hold true (for $L$), then (\ref{eq:cv-mu-ell-induction}) and (\ref{eq:Psi-N-ell-induction}) also hold true with $L$ replaced by $L'$.  

Since (\ref{eq:Psi-N-ell-induction}) holds true for $L$ and $\langle \Psi_{N,\ell'}, \Psi_{N,\ell} \rangle=0$ for all $\ell' \ne \ell$, we get
\begin{align} \label{eq:higher-ev-orthogonality-full}
\lim_{N\to \infty} \langle \Psi_{N,\ell'}, U_{N,j}^\dagger \Phi_{j,m}  \rangle =0
\end{align} 
for all $\ell'>L$, $j\in \{1,2,\dots,J\}$ and $1\le m \le r_j$, where $r_j$ is the number of eigenvalues of $\bH_j$ (counting multiplicity) among $\mu_1,\dots,\mu_L$, and $\Phi_{j,m}$ are the corresponding eigenvectors. We will show that 
\begin{align} \label{eq:higher-ev-orthogonality}
\lim_{k\to \infty} \lim_{N\to \infty} \langle \widehat  n_{j'}^k \Psi_{N,\ell'}, U_{N,j}^\dagger \Phi_{j,m}  \rangle =0
\end{align} 
for all $\ell'>L$, $j,j'\in \{1,2,\dots,J\}$ and $1\le m \le r_j$. Because of (\ref{eq:cv-X-0}) and (\ref{eq:higher-ev-orthogonality-full}), we just need to prove (\ref{eq:higher-ev-orthogonality}) when $j\ne j'$. In this case, (\ref{eq:higher-ev-orthogonality}) follows from the fact that  
$$\lim_{k\to \infty} \lim_{N\to \infty}\| \widehat  n_{j'}^k  U_{N,j}^\dagger \Phi_{j,m}\|=0.$$
To verify the latter convergence, we can use the fact that $U_{N,j}^\dagger \Phi_{j,m}$ condensates completely on $f_j$, in the sense that its density matrices satisfy
$$
\lim_{N\to \infty}\gamma_{U_{N,j}^\dagger \Phi_{j,m}}^{(k)} = |f_j^{\otimes k} \rangle \langle f_j^{\otimes k}|
$$
in trace class for all $k\in \mathbb{N}$, and the fact that $\langle f_j^{\otimes k}, f_{j'}^{\otimes k} \rangle \to 0$ as $k\to \infty$ when $j\ne j'$. Thus (\ref{eq:higher-ev-orthogonality}) holds true.

Now we come back to (\ref{eq:GSE-lambda-ell}). For every $j\in \{1,2,\dots,J\}$ and $\ell'>L$, we obtain from (\ref{eq:higher-ev-orthogonality}) and \eqref{eq:g-vanish} that 
$$
\lim_{k\to \infty} \lim_{N\to \infty}\left\langle  f\Big( \frac{\N_{+j}}{M}\Big) U_{N,j}  \widehat n_j^k \Psi_{N,\ell'}, \Phi_{j,m} \right\rangle =0
$$
for all $1\le m \le r_j$. Therefore, by the min-max principle in Lemma \ref{le:min-max-lower} and (\ref{eq:1-f-lambda}),  we have
\begin{align}  \label{eq:min-max-mu-ell+1}
&\liminf_{k\to \infty} \liminf_{N\to \infty} \left\langle  f\Big( \frac{\N_{+j}}{M}\Big) U_{N,j}  \widehat n_j^k \Psi_{N,\ell}, \bH_j  f\Big( \frac{\N_{+j}}{M}\Big) U_{N,j}\widehat n_j^k \Psi_{N,\ell} \right\rangle
\nn\\
&\ge \lambda_{j,\ell'}\mu_{r_j+1}(\bH_j).
\end{align}
Note that the condition $\mu_L<\mu_{L+1}$ implies that $\mu_{r_j}(\bH_j)<\mu_{r_j+1}(\bH_j)$, and also $\mu_{r_j+1}(\bH_j) \ge \mu_{L+1}$. Therefore, taking the sum over $j\in \{1,2,\dots,J\}$ in (\ref{eq:min-max-mu-ell+1}) and using (\ref{eq:GSE-lambda-ell}), we conclude that
\begin{align}  \label{eq:GSE-lambda-ell-fn}
&\quad \liminf_{N\to \infty} ( \mu_{\ell'}(H_N)-Ne_{\rm H}) \ge \sum_{j=1}^J \lambda_{j,\ell'}\mu_{r_j+1}(\bH_j) \ge \mu_{L+1}=\mu_{\ell'}
\end{align}
for all $\ell'\in \{L+1,\dots,L'\}$. Combining this with the upper bound (\ref{eq:mul-cond-upper-bound}), we conclude that \eqref{eq:cv-mu-ell-induction} holds true with $L$ replaced by $L'$. 

Now for every $j\in \{1,2,\dots,J\}$, let us assume that $\mu_{L+1},\dots,\mu_{L'}$ consist of $r'_j-r_j$ eigenvalues of $\bH_j$. Then $r'_j \ge r_j$ and $\sum_{j=1}^J r'_j=L'$. The condition $\mu_{L+1}=\mu_{L'}<\mu_{L'+1}$ implies that $\mu_{r_j}<\mu_{r_j+1}(\bH_j)=\dots=\mu_{r'_j}(\bH_j)<\mu_{r'_j+1}(\bH_j)$. Let $\Phi_{j,r_j+1}, \dots, \Phi_{j,r'_j}$ be the eigenvectors corresponding to the eigenvalues $\mu_{r_j+1}(\bH_j)=\dots=\mu_{r'_j}(\bH_j)$ of $\bH_j$ .  Note that for all $\ell'\in \{L+1,\dots,L'\}$, we have equality in (\ref{eq:min-max-mu-ell+1}). By the min-max principle in Lemma \ref{le:min-max-lower}, there is a subsequence of $\Psi_{N,\ell'}$ (still denoted by $\Psi_{N,\ell'}$ for short) satisfying
$$
\lim_{k\to \infty}\lim_{N\to \infty} f\Big( \frac{\N_{+j}}{M}\Big) U_{N,j}  \widehat n_j^k \Psi_{N,\ell'}=   \sum_{m=r_j+1}^{r'_j} \theta_{j,m}\Phi_{j,m}
$$
for all $j\in \{1,2,\dots,J\}$, where the $\theta_{j,m}$'s are complex numbers satisfying $\sum_{m=r_j+1}^{r'_j} |\theta_{j,m}|^2 =\lambda_{j,\ell'}$. Because of \eqref{eq:g-vanish} and (\ref{eq:cv-X-0}), we obtain the following analogue of (\ref{eq:final-cv-N-body-ground-state}),
$$
\lim_{N\to \infty} \left\| \Psi_{N,\ell'} - \sum_{j=1}^J \sum_{m=r_j+1}^{r'_j} \theta_{j,m} U_{N,j}^\dagger\Phi_{j,m} \right\| =0
$$
for all $\ell'\in \{L+1,\dots,L'\}$, which is the desired statement (\ref{eq:cv-eigenfunction-PsiN}) in Theorem \ref{thm:excitation-spectrum}.

Combining the latter convergence with the system (\ref{eq:Psi-N-ell-induction}) for $L$, we obtain the desired system the system (\ref{eq:Psi-N-ell-induction}) with $L$ replaced by $L'$, namely
\begin{align} \label{eq:Psi-N-ell'-induction}
&\big( \Psi_{N,1},\dots,\Psi_{N,L'} \big)^T \\
&\quad \quad = \mathfrak{A}_{L'} \big(U_{N,1}^\dagger\Phi_{1,1},\dots, U_{N,1}^\dagger\Phi_{1,n'_1},\dots, U_{N,J}^\dagger\Phi_{J,1},\dots,U_{N,J}^\dagger\Phi_{J,r'_J} \big)^T + \mathfrak{R}_{N,L'} \nn
\end{align}
where $\mathfrak{A}_{L'}$ is a $L'\times L'$ complex matrix independent of $N$, and $\|\mathfrak{R}_{N,L'}\|_{(\gH^N)^{L'}}\to 0$ as $N\to \infty$. This completes the proof of Theorem \ref{thm:excitation-spectrum}.
\end{proof}

\bigskip

\noindent\textbf{Acknowledgment.} We thank Nicolas Rougerie for inspiring discussions. The hospitality of the Institute for Mathematical Science of the National University of Singapore is gratefully  acknowledged.
%REFERENCES 

\end{document}